\documentclass[12pt,draftcls,onecolumn]{IEEEtran}

\usepackage{amsbsy,amsmath,amsfonts,amssymb,amsbsy,subfigure,verbatim}
\usepackage{bm,cite,graphicx,psfrag,pstricks,theorem,tikz,times,url,verbatim}
\usetikzlibrary{shapes,snakes,calendar,matrix,backgrounds,folding}
\usepackage[english]{babel}
\usetikzlibrary{arrows,automata}
\usetikzlibrary{calc} % for manimulation of coordinates
\usetikzlibrary{decorations.pathmorphing} % for snake lines

\newcommand{\bi}{\begin{itemize}}
\newcommand{\ei}{\end{itemize}}
\newcommand{\ben}{\begin{enumerate}}
\newcommand{\een}{\end{enumerate}}

\newcommand{\bc}{\begin{cases}}
\newcommand{\ec}{\end{cases}}
\newcommand{\bd}{\begin{description}}
\newcommand{\ed}{\end{description}}

\newcommand{\be}{\begin{equation}}
\newcommand{\ee}{\end{equation}}
\newcommand{\bea}{\begin{eqnarray}}
\newcommand{\eea}{\end{eqnarray}}

\newtheorem{thm}{Theorem}

\newtheorem{propos}{Proposition}

\newtheorem{algo}{Algorithm}

\theoremstyle{plain}
\newtheorem{remark}{Remark}

\newcommand{\nn}{\nonumber}

\graphicspath{%
     {./Figures/}
}

\begin{document}

\title{
Cross-layer estimation and control\\for Cognitive Radio:
\\
Exploiting Sparse Network Dynamics
}% with Fusion Center Uncertainty Feedback}

\author{Nicol\`{o}~Michelusi~and~Urbashi~Mitra\\(Invited Paper)
\thanks{%%{-5mm}\newline 
This research has been funded in part by the following grants:
ONR N00014-09-1-0700, AFOSR FA9550-12-1-0215, DOT CA-26-7084-00, NSF CCF-1117896, NSF CNS-1213128, NSF CCF-1410009, NSF CPS-1446901.
}
\thanks{N. Michelusi and U. Mitra are with the Dept. of Electrical Engineering, University of Southern California. email: michelus@usc.edu.}
\vspace{-20mm}
}

\maketitle

%Controlled spectrum sensing and scheduling under resource constraints
%Joint Control and Compressed Sensing for Dynamic Spectrum Access in Agile Wireless Networks
%Dynamic Spectrum Sensing-Scheduling in Agile Networks with Compressed Belief Information
%A Cross-Layer Framework for Joint Control and Distributed Sensing in Agile Wireless Networks

\begin{abstract}
In this paper, a cross-layer framework to jointly optimize spectrum sensing and scheduling in resource constrained agile wireless networks is presented. A network of secondary users (SUs) accesses portions of the spectrum left unused by a network of licensed primary users (PUs). A central controller (CC) schedules the traffic of the SUs, based on distributed compressed measurements collected by the SUs. Sensing and scheduling are jointly \emph{controlled} to maximize the SU throughput, with constraints on PU throughput degradation and SU cost.  The sparsity in the spectrum dynamics is exploited: leveraging a prior spectrum occupancy estimate, the CC needs to estimate only a residual uncertainty vector via sparse recovery techniques.  The high complexity entailed by the POMDP formulation is reduced by a low-dimensional belief representation via minimization of the Kullback-Leibler divergence. It is proved that the optimization of sensing and scheduling can be decoupled.   A \emph{partially myopic} scheduling strategy is proposed for which structural properties can be proved showing that the myopic scheme allocates SU traffic to likely idle spectral bands.   Simulation results show that this framework balances optimally the resources between spectrum sensing and data transmission. This framework defines sensing-scheduling schemes most informative for network control, yielding energy efficient resource utilization.
\end{abstract}
\vspace{-5mm}
\section{Introduction}
The recent proliferation of mobile devices has been exponential in number as well as heterogeneity \cite{CISCO}.
 As mobile data traffic is expected to grow 13-fold, and machine-to-machine traffic will experience a 24-fold increase from 2012 to 2017 \cite{CISCO}, tools for the design and optimization of \emph{agile} wireless networks
is of significant interest \cite{pcast}. Furthermore, network design needs to explicitly consider the resource constraints typical of wireless systems.
These resource constraints will impact the acquisition of network state information, which is essential for network control.

In this paper, we consider a wireless network composed of a licensed network of primary users (PUs) dynamically accessing a spectrum with $F$ frequency bands,
and an agile network of secondary users (SUs) which opportunistically attempt to access the portion of the spectrum left unused by the PUs \cite{Meng}.
The spectrum occupancy is inferred by a central controller (CC), by aggregating compressed spectrum measurements
collected   in a \emph{distributed} fashion by the SUs,
and by overhearing feedback signaling from the PUs.
Accordingly, the CC allocates the traffic of the SUs across the spectrum bands.
Joint sensing-scheduling policies are defined,
   so as to  maximize the SU throughput, under constraints on the throughput degradation caused to the PUs and
   on the sensing-transmission cost incurred by the SUs. 
   
The contributions of this paper are as follows.
We propose a  framework 
which captures the interplay between \emph{sensing and scheduling}, 
by trading off the cost of acquisition of network state information and the overall network performance.
\emph{Spectrum sensing} is done by collecting \emph{compressed} spectrum measurements from distributed SUs
and local feedback at the CC. based on it, \emph{spectrum scheduling} decisions are done.
This is in contrast to standard formulations based on partially observable Markov decision processes (POMDPs) \cite{POMDP},
where observations are \emph{passively} generated by control actions, rather than \emph{actively controlled} via sensing.
We provide a motivational example for the case of a single spectrum band and noiseless sensing in Sec. \ref{motivation}, 
which highlights the need for adaptivity in a cross-layer and resource constrained environment,
and then extend the model to the general case. 
For the general case, in Sec. \ref{optimization},
 we show that the joint sensing-scheduling policy can be optimized via dynamic programming (DP); we prove the optimality of a two-stage decomposition,
 which exploits the sufficient statistics that drive the decision making process (Theorem \ref{suffstat}), and
 allows one to decouple the optimization of sensing and scheduling (Algorithm \ref{aDP}).
Additionally, in order to reduce the huge action space in the spectrum scheduling phase, 
we propose a \emph{partially myopic} scheduling scheme, where the total traffic of the SUs is determined optimally via DP, whereas
 the allocation of the resulting total budget across frequency bands is determined via a myopic maximization of the instantaneous trade-off between PU and SU throughput.
We prove structural properties of the partially myopic scheduling scheme, showing that it effectively allocates the SU traffic to the spectrum bands more likely to be idle, thus minimizing interference  to the PUs and maximizing SU throughput,
and that it can be solved efficiently using standard convex optimization tools (Theorem \ref{thm2}).

In order to tackle the high complexity of the DP algorithm \cite{Bertsekas2005},
in Sec. \ref{complereduction} we propose complexity reduction techniques.
We employ a compact state space representation by projecting the belief onto a low-dimensional manifold via the minimization
of the Kullback-Leibler divergence (KLD, Theorem \ref{thm1}).
 Based on the \emph{compressed} belief,
we design adaptive compressive sensing (CS) schemes, which effectively exploit the  \emph{sparse} network dynamics typical of wireless networks.
In the spectrum sensing context analyzed in this paper, only few PUs join or leave the spectrum at any time,
so that the spectrum occupancy state exhibits sparse time variations.
Therefore, leveraging the estimate of the spectrum occupancy state in the previous slot,
only a \emph{sparse} residual uncertainty vector needs to be estimated, and
few measurements suffice to drive scheduling decisions.
Additionally, such representation allows us to design a state estimator based on sparse recovery algorithms.
  Although the focus of this paper is on spectrum sensing in agile wireless networks,
this framework can be generalized to more general networked systems, where the state of the system is a collection of \emph{features}, rather than spectrum bands (\emph{e.g.}, buffer state of all wireless nodes, or local channel quality), which evolve sparsely over time. These state features can be tracked
by collecting a few compressed measurements  via distributed sensing,
enabling more informed network control.
\vspace{-5mm}
\subsection{Related work}
There is significant prior work on cognitive radio and compressed sensing (CS); we have focused on the literature that is most relevant to our current problem framework.  Centralized schemes for the tracking of sparse time-varying processes have been examined in \cite{Schniter,Donoho,Asif} and distributed CS has been studied in \cite{Baron,Mota} for \emph{static} signals. In contrast to these two veins, we study distributed CS for time-varying signals.  Performance guarantees for recursive reconstruction of sparse signals under the assumption of slow support changes is studied in \cite{Vaswani}; however, joint sensing and control is not examined.  Recovery of static binary sparse signals via CS has been investigated in~\cite{Stojnic,Nakarmi}. Compressive spectrum sensing has been studied in~\cite{Meng}, for a static setting, and in  \cite{Bagheri}, for a dynamic setting with noiseless measurements, but without scheduling.  We do not focus on recovery guarantees herein, but rather embed CS into a control framework wherein the number of measurements is adapted based on prior information in order to drive traffic scheduling decisions.

Active sensor scheduling and adaptation \cite{Hero} encompass applications such as target tracking \cite{Atia,Krishnamurthy},
 physical activity detection \cite{Zois}, and sequential hypothesis testing \cite{Javidi}. All these prior works including ours \cite{MicheTSP1,MicheTSP2,MicheICASSP} assume that the underlying state is given by nature and is not controlled. In contrast, in this work, states are affected by scheduling decisions, via interference and collisions generated by the SUs to the PUs, and we design joint \emph{controlled} sensing, estimation and scheduling schemes in wireless networks, which account for the cost of acquisition of state information and its impact on the overall network performance. 
 
 Complexity reduction of  POMDPs via exponential family principal components analysis enables planning on a small dimensional manifold in \cite{Thrun}.  Model reduction of complex Markov chain models using the KLD  as a metric is investigated in \cite{Meyn}. In contrast, we develop a belief compression method based  on Neyman-Pearson formulation of the compressive spectrum sensing problem. Our scheme captures relevant features of the dynamic spectrum access problem, without having to learn key statistics.  As in \cite{Meyn}, the KLD measure is also used to project the true belief onto the low-dimensional manifold.

In this work, we assume that the PUs employ a retransmission process, which induces structure in the PU signal. This structure has been exploited in \cite{LevoratoIT} to design adaptive SU access techniques, and in \cite{mychain} to design smart interference cancellation techniques that exploit redundancy introduced by retransmission. In this work, instead, we exploit the structure in the signal as a result of sparse network dynamics, to design compressive spectrum sensing techniques and sparse recovery schemes. We extend the model studied in \cite{MicheGsip} to include a more general traffic model for the SU network, and propose a low-complexity solution based on the aforementioned partially myopic scheduling scheme.

This paper is organized as follows. In Sec. \ref{motivation}, we provide an example which motivates the need for adaptivity in a cross-layer and resource constrained environment, for the case with a single frequency band and noiseless sensing. In Sec.~\ref{model}, we present the system model for the general case with multiple frequency bands and noisy sensing. In Sec.~\ref{optimiz}, we present the optimization problem and, in Sec. \ref{optimization}, the proposed optimization techniques. In Sec. \ref{complereduction}, we present techniques for the complexity reduction based on belief approximation via KLD minimization and sparse recovery algorithms. In Sec.~\ref{sec:numres}, we present numerical results, and, in Sec.~\ref{conclu}, we conclude the paper. The proofs of the analytical results are provided in the Appendix.

\section{Motivation: single frequency band and noiseless sensing}
\label{motivation}
In this section, we provide an example which motivates  the need for adaptivity in a cross-layer and resource constrained environment,
by comparing the performance achieved by 
non-adaptive sensing strategies (Sec.\ref{nasensing}), with that achieved by adaptive schemes (Sec.\ref{asensing}).
In particular, we focus on the special case of a single frequency band and noiseless sensing.
Consider a network of $N_S$ SUs with sensing capability, which attempt to access a licensed channel (single frequency band), represented in Fig. \ref{TOYEX}.
Herein, for mathematical convenience, we use the approximation $N_S{\to}\infty$ to derive the transition probabilities and performance of the system.
The following discussion can be generalized to $N_S{<}\infty$.
  The channel occupancy state in slot $k$ is denoted as 
$ b_{k}\in\{0,1\}$, where $ b_{k}=0$ if the channel is idle and $ b_{k}=1$ if it is occupied by a PU.
%The PUs occupy the spectrum according to a stochastic process, so that $\mathbf b_{k}$ is time-varying.

The SUs opportunistically access the spectrum based on the traffic decision $ r_k$ broadcasted by the CC.
 Given $ r_{k}$, each SU, assumed to be backlogged,  transmits data independently with probability $ q_{k}= r_{k}/N_S$,
 incurring the transmission cost $c_{TX}$; otherwise, the SU remains idle, incurring no cost.
We employ a collision channel model, \emph{i.e.}, if more than one terminal (either SUs or PUs) transmits on the same channel, those packets 
cannot be decoded correctly at the corresponding receiver and are lost. Otherwise, if one and only one user transmits,
then the transmission is successful with probability $1{-}\rho_S$ (for the SU) and $1{-}\rho_P$ (for the PU). 
 %, where $\rho_S$ and $\rho_P$ are the respective transmission failure probabilities.
 This collision model represents a worst-case scenario, and thus provides performance guarantees. 
The value $ r_{k}=1$ maximizes the throughput for the SUs \cite{Abramson}, and any larger
value $ r_{k}>1$ degrades both the PU and SU throughputs, and incurs larger energy cost for the SUs.
We thus restrict $r_{k}$ to take values in $[0,1]$.

\begin{figure}[t]
\centering
%{5mm}
\includegraphics[width = .6\linewidth,trim = 0mm 4mm 10mm 9mm,clip=false]{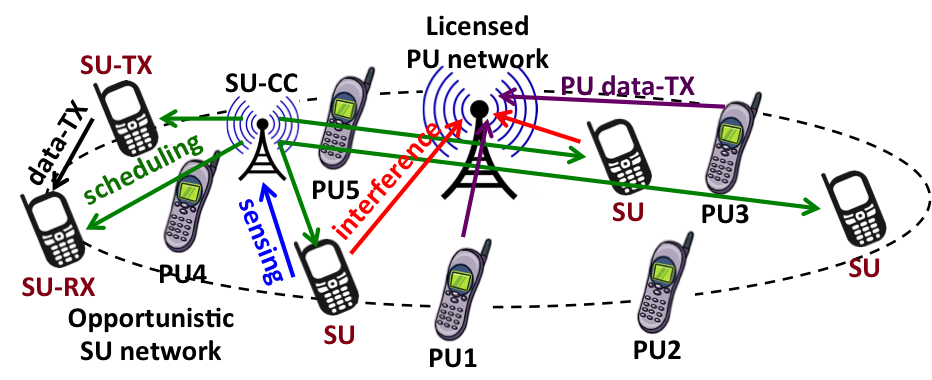}
\caption{Licensed network of PUs and opportunistic network of SUs. The SU-CC receives spectrum measurements and controls the SU network accordingly. SU transmissions generate interference to the PU network.}
\label{TOYEX}
\vspace{-9mm}
\end{figure}

The success probability for the PUs as a function of $ b_{k}$ and $ r_{k}$ is given by 
\begin{align}
 \label{PsuccPU}
&P_{succ}^{(P)}\left(  b_{k}, r_{k}\right)\triangleq  b_{k}(1-\rho_P)e^{- r_{k}},
\end{align}
where
 the probability of no collisions from the SUs satisfies
$(1{-} r_{k}/N_S)^{N_S}{\to}e^{- r_{k}}$ for $N_S{\to}\infty$.
Similarly, the probability of successful transmission for the SU system
 is given by
 \begin{align}
 \label{PsuccSU}
P_{succ}^{(S)}\left( b_{k}, r_{k}\right)\triangleq
(1- b_{k})(1-\rho_S) r_{k}e^{- r_{k}},
\end{align}
where the probability that one and only one SU transmits satisfies
$r_{k}\left(1- r_{k}/N_S\right)^{N_S-1}\to r_{k}e^{- r_{k}}$,
 and, if the channel is occupied by one PU, the transmission fails due to collisions.\footnote{Note that
 the analysis under  the asymptotic assumption $N_S{\to}\infty$ yields a good approximation even when $N_S$ is finite,   \emph{e.g.}, $N_S\simeq 10$.
 For instance, if $\rho_P=0$ and $r_k=1$ in (\ref{PsuccPU}), or $\rho_S=0$ in  (\ref{PsuccSU}),
 we obtain $P_{succ}^{(P)}\left(1,1\right)=P_{succ}^{(S)}\left(0,1\right)\simeq 0.368$ in the asymptotic case 
 and $P_{succ}^{(P)}\left(1,1\right)=P_{succ}^{(S)}\left(0,1\right)\simeq 0.349$ when $N_S=10$.}
 
 The PUs implement a retransmission mechanism in case of  transmission failure.
Retransmissions are performed in the same channel, in the next slot.
If the transmission is successful, then the PU occupying the channel either
has a new data packet to transmit in the next slot,
with probability $\theta$, or leaves the spectrum idle.
An idle channel is occupied by a new PU with probability $\zeta\in(0,1)$ and it remains idle otherwise.
Therefore, the state $b_{k}\in\{0,1\}$
is a two-state Markov chain, whose transition probabilities depend on the allocated SU traffic, $r_{k}\in[0,1]$.
The transition probability 
from state $ b_{k}{=}b$ to $ b_{k+1}{=}b^\prime$, given
$r_{k}{=}r$, denoted 
 as $P_{B}(b^\prime|b,r){=}\mathbb P( b_{k+1}{=}b^\prime| b_{k}{=}b, r_{k}{=}r)$, is given~by
 \begin{align*}
\!\!\!\!P_{B}(1|b,r)=(1-b)\zeta
{+}\left[b{-}(1{-}\theta)(1{-}\zeta)P_{succ}^{(P)}\left(b,r\right)\right]
\end{align*}
and $P_{B}(0|b,r)=1-P_{B}(1|b,r)$.
In fact, the channel is occupied in the next slot if and only if one of the following events occur:
the PU transmits successfully and it has a new data packet to transmit, with probability $\theta$;
 a new PU arrives, with probability $\zeta$;
or the transmission of the PU is unsuccessful and thus a retransmission is required.
%We define the aggregate expected instantaneous throughput for the SU and PU systems, respectively, given $ b_{k}$ and $ r_k$, as
% \begin{align}
%&T_{X}( b_{k}, r_k)=P_{succ}^{(X)}\left( b_{k}, r_{k}\right),
 % \end{align}
 % where $X\in\{S,P\}$ is a label indicating SU or PU, respectively.
 %{-5mm}
 \vspace{-5mm}
 \begin{remark}
The retransmission protocol implemented by the PUs can also be exploited
 by leveraging the redundancy of the retransmission process, using a technique termed \emph{chain decoding} 
 \cite{mychain} to remove the interference of the PU signal over the retransmission windows of the PU.
 In this paper, we assume slot-by-slot decoding, so that the redundancy of the PU retransmission protocol
 is not exploited for interference cancellation. The extension is left for future work.
\end{remark}
\vspace{-5mm} 
\subsection{Non-adaptive spectrum sensing}
\label{nasensing}
The SU traffic $r_k$ is scheduled based on spectrum measurements collected by the SUs in a distributed fashion.
Consider a scheme where the SUs collect and report to the CC noiseless spectrum measurements
 at the beginning of slot $k$, with probability $\alpha=\psi /N_S$ independently in each slot, incurring the sensing-transmission cost $c_{S}$,
and they remain idle otherwise, incurring no cost.
The parameter $\psi\in[0,1]$ denotes the average SU \emph{sensing traffic}.
The SUs share a control channel to report their measurements, resulting in packet losses
if more than one SU transmits on the same channel.
The probability that the CC collects the spectrum measurement is thus given by $p_S=N_S \alpha(1-\alpha)^{N_S-1}\to \psi e^{-\psi}$ (for $N_S\to\infty$).

Assume that the SUs are not allowed to cause any degradation to the PU system. Then,
the SU traffic  is $r_k=r\in [0,1]$ in those slots where the channel is detected by the CC to be idle,
otherwise no traffic is allowed (in order to not interfere with the PUs).
In particular, if no measurement is collected, no SU transmissions are allowed, due to the uncertainty in the current
channel state.
The average long-term sensing and transmission cost incurred by the SUs, and the SU and PU throughputs are given by
\begin{align}
&\bar C_{sensing}(\psi,r)=\psi c_S,
&\bar C_{sched}(\psi,r)=\pi_P(0)\psi e^{-\psi} rc_{TX},
\\
&\bar T_S(\psi,r)=(1-\rho_S)\pi_P(0)\psi e^{-\psi}  r e^{-r},
&\bar T_P(\psi,r)=(1-\rho_P)\pi_P(1),
\end{align}
where $\pi_P(0)$ and $\pi_P(1)$ are, respectively, the steady-state probabilities of the channel being idle and occupied,
given by
\begin{align}
\pi_P(0)&=\frac{P_{B}(0|1,0)}{P_{B}(0|1,0)+P_{B}(1|0,r)}
%\nonumber\\&
=
\frac{(1-\theta)(1-\zeta)(1-\rho_P)}{(1-\theta)(1-\zeta)(1-\rho_P)+\zeta}
\end{align}
and $\pi_P(1)=1-\pi_P(0)$.
In fact, sensing is done independently in each slot, incurring the expected sensing cost $\psi c_S$.
If the measurement is received successfully (with probability $\psi e^{-\psi}$) and
 the channel is detected to be idle (with steady-state probability $\pi_P(0)$),
 then the data transmission cost $rc_{TX}$ is incurred in the scheduling phase, and the expected throughput achieved is $r e^{-r}$.
 
 We want to determine $(\psi^*,r^*)$ such that
\begin{align}
\label{opaasd}
(\psi^*,r^*)=\arg\max_{\psi,r}\bar T_S(\psi,r)
\ \text{s.t.}\ \bar C(\psi,r)\leq \bar C^{\max},
\end{align}
where we have defined the sensing-transmission cost $\bar C(\psi,r)\triangleq \bar C_{sensing}(\psi,r)+\bar C_{sched}(\psi,r)$, and
 $\bar C^{\max}\leq c_S+\pi_P(0)e^{-1}c_{TX}$ (achieved with $\psi=r=1$).
\emph{The
 above optimization problem allows us to define the joint sensing-scheduling strategy that balances optimally between the cost of acquisition of state information via distributed
sensing and the overall network goal of maximizing the SU throughput and, at the same time, avoiding interference to the PUs.}

Since both $\bar C(\psi,r)$ and $\bar T_S(\psi,r)$ are increasing functions of $\psi\in[0,1]$ and $r\in[0,1]$, under the optimal strategy
we have $\bar C(\psi^*,r^*)=\bar C^{\max}$, yielding the optimal $r$ as a function of $\psi$,
\begin{align}
&r(\psi)=\frac{\bar C^{\max}-\psi c_S}{\pi_P(0)\psi c_{TX}}e^{\psi},
\end{align}
where $\psi^*\leq\min\left\{\frac{\bar C^{\max}}{c_S},1\right\}$. 
 %We thus obtain
%\begin{align}&\bar T_S(\psi,r(\psi))=(1-\rho_S)\frac{\bar C^{\max}-\psi c_S}{c_{TX}} e^{-r(\psi)},\end{align}
 Hence $\psi^*$ can be determined as
\begin{align}
\psi^*=\underset{\psi\in \left[0,\min\left\{\frac{\bar C^{\max}}{c_S},1\right\}\right]}{\arg\max} \bar T_S(\psi,r(\psi)),
\end{align}
by exhaustive search.
When $\bar C^{\max}\ll c_S$, hence $\psi\ll 1$, we approximate $e^{\psi}\simeq 1$,
thus obtaining 
\begin{align}
\bar T_S(\psi,r(\psi))&\simeq(1-\rho_S)\frac{\bar C^{\max}-\psi c_S}{c_{TX}} e^{-\frac{\bar C^{\max}-\psi c_S}{\pi_P(0)\psi c_{TX}}}
%\nonumber\\&
\triangleq
\bar T_S^{(up)}(\psi,r(\psi)),
\end{align}
which represents an upper bound to $\bar T_S(\psi,r(\psi))$ for the general case. 
This upper bound can be optimized in closed form, %yielding the optimal $\psi^*$ and $r^*$,
yielding the upper bound optimizing $\psi^*$ and $r^*$,
\begin{comment}
 By solving $\frac{\mathrm d\bar T_S^{(up)}(\psi,r(\psi))}{\mathrm d\psi}{>}0$,
it can be shown that $\bar T_S(\psi,r(\psi))$ is an increasing function of $\psi$ if and only if 
\begin{align}
\psi<\frac{2 \bar C^{\max}/c_S}{1+\sqrt{1+4\pi_P(0)c_{TX}/c_S}},
\end{align}
thus yielding the optimal $\psi^*$ and $r^*$,
\end{comment}
\begin{align}
&\psi^*%=\min\left\{\frac{2 \bar C^{\max}}{c_S[\sqrt{1+4\pi_P(0)c_{TX}/c_S}+1]},\frac{\bar C^{\max}}{c_S},1\right\}\nonumber\\&
=\min\left\{\frac{2 \bar C^{\max}/c_S}{1+\sqrt{1+4\pi_P(0)c_{TX}/c_S}},1\right\},
%\nonumber\\&
%=
%\left\{
%\begin{array}{ll}
%1 & \bar C^{\max}\geq \frac{1}{2}c_S[\sqrt{1+4\pi_P(0)c_{TX}/c_S}+1]\\
%\frac{2 \bar C^{\max}}{c_S[\sqrt{1+4\pi_P(0)c_{TX}/c_S}+1]} & \text{otherwise}
%\end{array}
%\right..
\\
&r^*=r(\psi^*)
=\frac{\bar C^{\max}-\psi^* c_S}{\pi_P(0)\psi^* c_{TX}}e^{\psi^*}.
%\\&\nonumber
%\frac{c_S}{\pi_P(0)c_{TX}}
%\max\left\{\frac{\sqrt{1+4\pi_P(0)c_{TX}/c_S}-1}{2},\frac{\bar C^{\max}}{c_S}-1\right\}.
\end{align}
\begin{comment}
The fraction of total cost budget spent for spectrum sensing is thus
\begin{align*}
\frac{\bar C_{sensing}(\psi,r)}{\bar C(\psi,r)}
=
\min\left\{\frac{2}{1+\sqrt{1+4\pi_P(0)c_{TX}/c_S}},\frac{c_S}{\bar C^{\max}}\right\},
\end{align*}
which represents the optimal balance between spectrum sensing and SU data transmission under this non-adaptive sensing scheme.
\end{comment}
\vspace{-10mm}
\subsection{Adaptive spectrum sensing-scheduling}
\label{asensing}
The above non-adaptive sensing strategy does not provide the best performance possible due to its static nature. Indeed, 
it may be beneficial to adapt the sensing strategy over time, \emph{i.e.}, by selectively sensing
the channel state based on the prior channel information, in order to make the best use of the scarce resources available to the SUs.
We now demonstrate the importance of using adaptive sensing schemes to optimize the performance of the system,
as a means to effectively
cope with the cost of acquisition of state information for network control.

Thus, we consider the scenario where the sensing traffic $\psi$ is adapted over time.
We denote the belief state at the CC as $(b,\tau)$, where
$\tau\geq 0$ denotes the number of slots since the last measurement was collected,
and $b\in\{0,1\}$ denotes the last channel state detected.
 For instance, $b_k=0$, $\tau_k=1$
 denotes that the spectrum was detected as idle in slot $k-1$.
  Let $\psi(b,\tau)\in[0,1]$ be the sensing traffic, \emph{i.e.}, the expected number of measurements collected by the network of SUs, when the state is $(b,\tau)$,
  so that the probability that a measurement is successfully collected is given by $p_S(b,\tau)=\psi(b,\tau)e^{-\psi(b,\tau)}$.
 We denote the \emph{prior steady-state distribution} (before sensing) that the belief is $(b,\tau)$ as $\pi(b,\tau)$;
 similarly, we denote
 the \emph{posterior  steady-state distribution} (after sensing)  that the belief is $(b,\tau)$ as $\hat\pi(b,\tau)$.
 The steady-state equations relating the prior to the posterior steady-state probabilities are given by
 \begin{align}
 \label{q1}
& \hat\pi(b,\tau)=\pi(b,\tau)(1-p_S(b,\tau)),\ \tau>0,
\\&
 \hat\pi(b,0)=\sum_{b^{\prime}\in\{0,1\}}\sum_{\tau=1}^{\infty}\pi(b^\prime,\tau)p_S(b^\prime,\tau)\mathbb P^{(\tau)}(b|b^\prime),
  \label{q2}
 \end{align}
 where $\mathbb P^{(\tau)}(b^\prime|b)$ is the $\tau$-step probability of transition of the channel from state $b^\prime$ to state $b$.
 In fact, the posterior belief $(b,\tau)$ for $\tau{>}0$ is reached from the prior belief $(b,\tau)$ if no measurement is successfully collected at the CC.
 On the other hand, the posterior belief $(b,0)$ is reached if the measurement is collected at the CC and the channel state $b$ is detected.
 
Similarly,  the steady-state equations relating the posterior to the prior steady-state probability in the next slot are given by
 \begin{align}
 \pi(b,\tau)=\hat\pi(b,\tau-1),\ \forall b\in \{0,1\},\forall \tau\geq1.
  \label{q3}
 \end{align}
 In fact, since we are moving to the next slot, the information about the last state detected becomes outdated by one more slot.
 By solving the system of equations (\ref{q1}-\ref{q3}), we obtain
  \begin{align}
&
\!\!\!\pi(1,\tau)=\frac{f(1)}{
\sum_{b\in\{0,1\}}f(b)\sum_{\tau=1}^{\infty}\prod_{i=1}^{\tau-1}(1-p_S(b,i))
}\prod_{i=1}^{\tau-1}(1-p_S(1,i)),\ \tau>1,
 \\&
\!\!\!\pi(0,\tau)=\frac{f(0)}{
\sum_{b\in\{0,1\}}f(b)\sum_{\tau=1}^{\infty}\prod_{i=1}^{\tau-1}(1-p_S(b,i))
}\prod_{i=1}^{\tau-1}(1-p_S(0,i)),\ \tau>1,
 \end{align}
 where we have defined
 \begin{align}
& f(b)=\sum_{\tau=1}^{\infty}\prod_{i=1}^{\tau-1}(1-p_S(1-b,i))p_S(1-b,\tau)\mathbb P^{(\tau)}(b|1-b).
 \end{align}
We thus obtain
\begin{align}
&\bar C(\psi,r)=
\sum_{b\in\{0,1\}}\sum_{\tau=1}^{\infty}\pi(b,\tau)\psi(b,\tau)c_S
+
\hat\pi(0,0)rc_{TX},
\\&
\bar T_S(\psi,r)=(1-\rho_S)\hat\pi(0,0)re^{-r},%=(1-\rho_S)\pi(0,1)re^{-r},
%=\sum_{b\in\{0,1\}}\sum_{n=1}^{\infty}\pi(b,n)\psi(b,n)c_S+\pi(0,1)rc_{TX}
%[\pi(0,1)+\pi(1,1)]c_S+\pi(0,1)rc_{TX},
\end{align}
where we have used the fact that data transmission occurs only when the spectrum is detected to be idle (state $(0,0)$),
with cost $rc_{TX}$ and instantaneous throughput $(1-\rho_S)re^{-r}$.

The goal is to define jointly the sensing-scheduling  policy $(\psi^*,r^*)$ solving (\ref{opaasd}).
The optimal policy can be determined via dynamic programming. For simplicity and for the sake of exposition, here we evaluate the performance of an heuristic adaptive sensing policy such that
$\psi(b,\tau)=\psi(b)$, hence $p_S(b,\tau)=p_S(b)$, \emph{i.e.}, the sensing probability is only adapted to the value of the last state detected, rather than the delay parameter $\tau$.
In this case, we obtain
\begin{align}
&\pi(b,\tau)=\frac{f(b)}{\frac{f(0)}{p_S(0)}+\frac{f(1)}{p_S(1)}}(1-p_S(b))^{\tau-1},\ \tau\geq1,b\in\{0,1\},
 \end{align}
 and therefore
 \begin{comment}
 where
 \begin{align}
f(b)=\sum_{\tau=1}^{\infty}(1-p_S(1-b))^{\tau-1}p_S(1-b)\mathbb P^{(\tau)}(b|1-b).
 \end{align}
Defining the one-step transition probability matrix
 \begin{align}
 \mathbf P=\left[
 \begin{array}{cc}
 P_{B}(0|0,0) &  P_{B}(1|0,0)
 \\
 P_{B}(0|1,0) &  P_{B}(1|1,0)
 \end{array}
 \right],
 \end{align}
we can rewrite $f(b)$ as
   \begin{align}
&f(b)=\sum_{\tau=1}^{\infty}(1-p_S(1-b))^{\tau-1}p_S(1-b)\mathbf e_{1-b}^T\mathbf P^\tau\mathbf e_{b}
%\nonumber\\&
=p_S(1-b)\mathbf e_{1-b}^T\mathbf P\left(\mathbf I-(1-p_S(1-b))\mathbf P\right)^{-1}\mathbf e_{b},
 \end{align}
 where $\mathbf e_0=[1,0]^T$ and $\mathbf e_1=[0,1]^T$. 
   Thus, we obtain
    \end{comment}
  \begin{align*}
\bar C(\psi,r)=
c_S\left[\hat\pi(0,0)e^{\psi(0)}+\hat\pi(1,0)e^{\psi(1)}\right]
+\hat\pi(0,0)rc_{TX}.
%=\sum_{b\in\{0,1\}}\sum_{n=1}^{\infty}\pi(b,n)\psi(b,n)c_S+\pi(0,1)rc_{TX}
%[\pi(0,1)+\pi(1,1)]c_S+\pi(0,1)rc_{TX},
\end{align*}
By optimizing numerically the SU throughput $\bar T_S(\psi,r)$ with respect to $(\psi(0),\psi(1),r)$,
we obtain the plot in Fig. \ref{fig:example}.a, where we also plot the non-adaptive sensing policy (unless otherwise stated, the parameters are given as in Sec. \ref{sec:numres}).
We observe that the adaptive scheme achieves twice as much SU throughput as the non-adaptive one, for low values of the cost budget; the lower cost budget  is typical for wireless systems.
In Fig. \ref{fig:example}.b, we plot the ratio between the sensing cost $\bar C_{sensing}$ and the total budget $\bar C^{\max}$.
For both schemes, more than $65\%$ of the resources is spent for sensing, and consequently less than $35\%$
 is used for SU data transmission.

\begin{figure}[t]
    \centering
    \subfigure[SU throughput versus total cost of sensing-scheduling.]
    {
\includegraphics[width=.45\linewidth,trim = 1mm 1mm 1mm 1mm,clip=true]{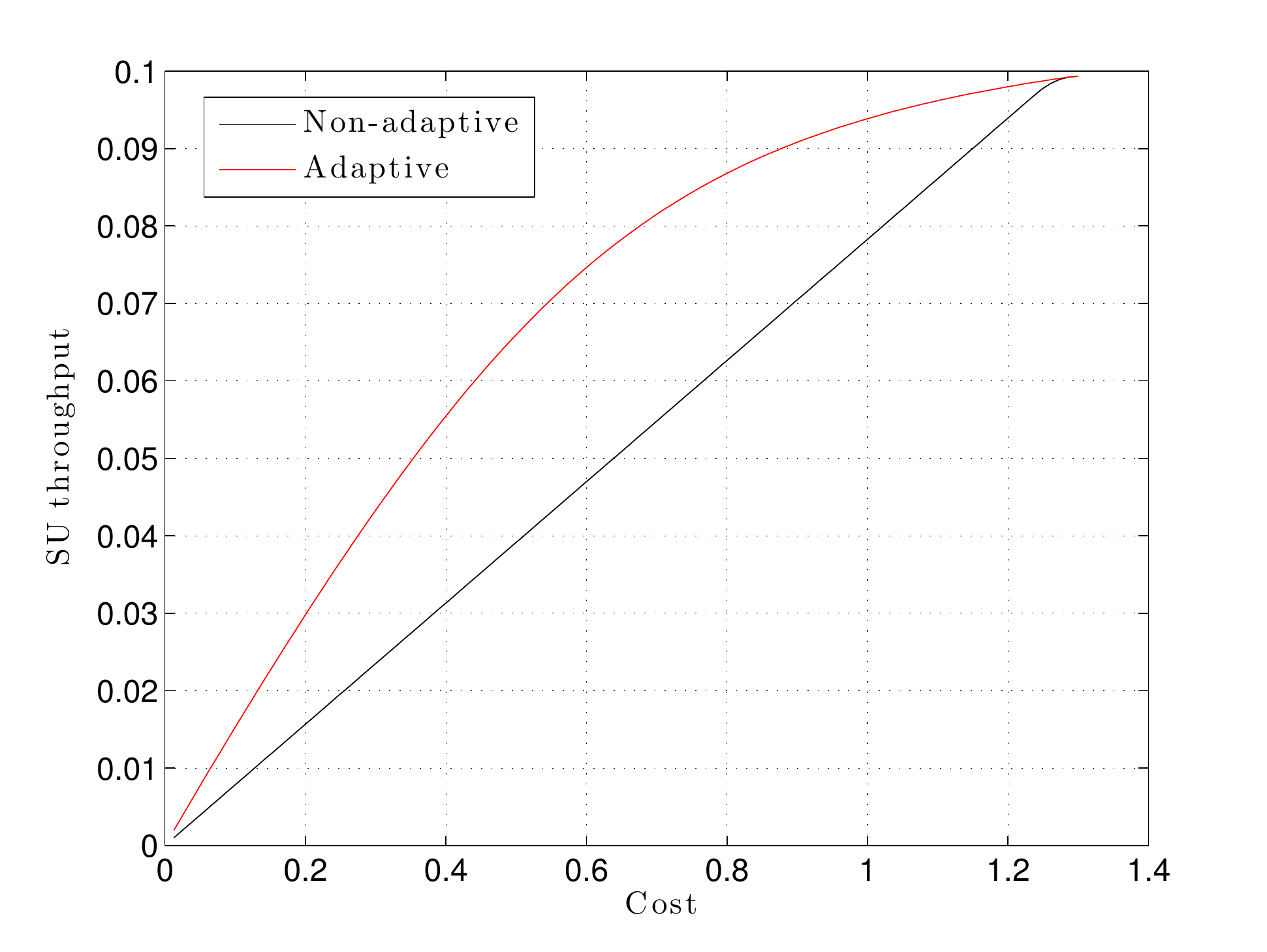}
}
    \subfigure[Ratio between sensing cost and total cost
 versus total cost of sensing-scheduling.]
    {
\includegraphics[width=.45\linewidth,trim = 1mm 1mm 1mm 1mm,clip=true]{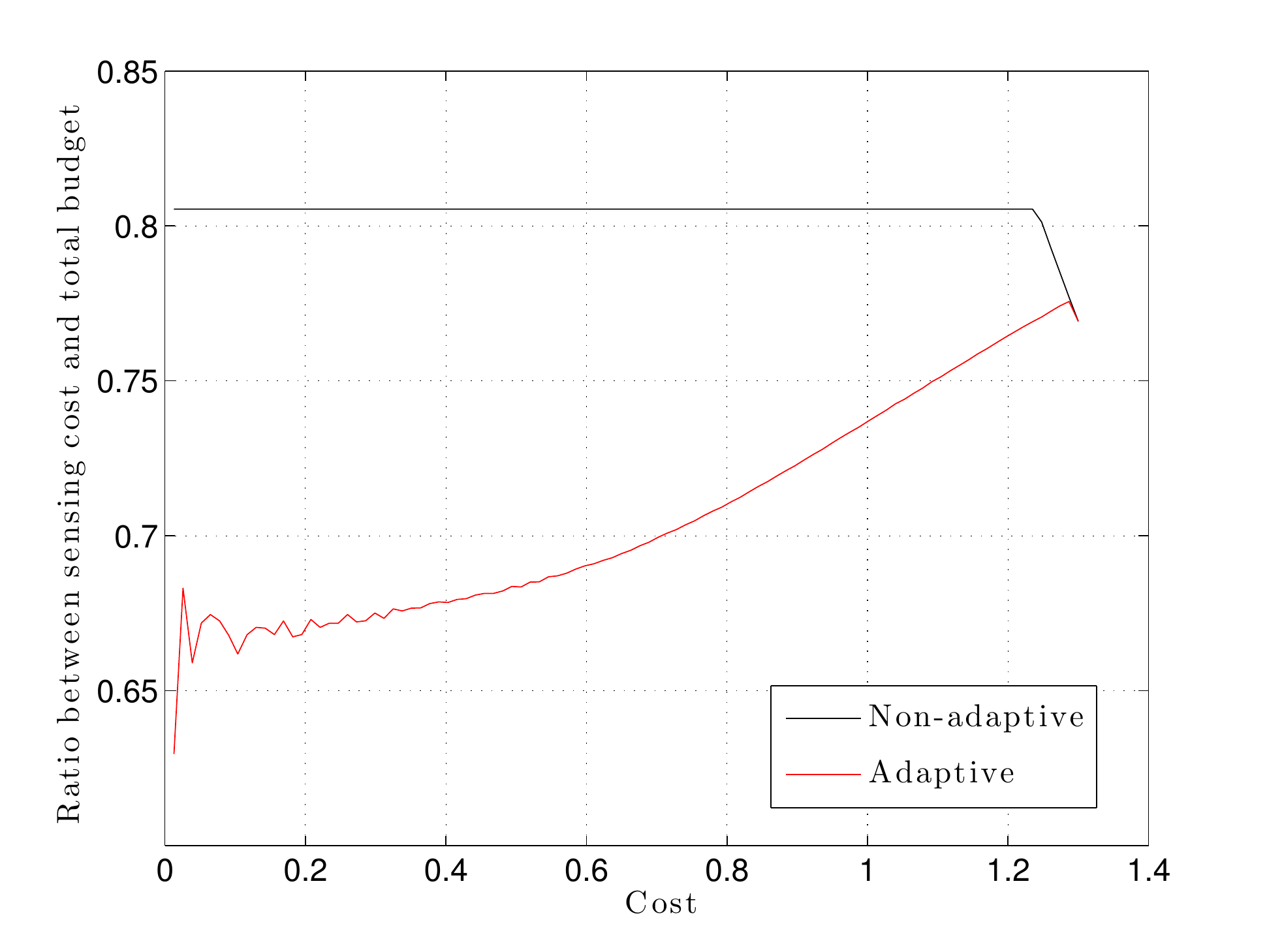}
}
\caption{}
    \label{fig:example}
    \vspace{-9mm}
\end{figure}

This example demonstrates the importance of taking into account the cost of acquisition of state information for network control,
and the importance of using adaptive sensing schemes to optimize the performance of the system.
In the next section, we investigate the more general case with noisy \emph{compressed} measurements collected by the SUs,
$F\geq 1$ frequency bands, $B\geq 1$ control channels employed by the SUs to report their measurements, and feedback from the PUs.

\section{System Model: multiple frequency bands and noisy sensing}\label{model}
In this section, we extend the model considered in the previous section to the more practical setting with multiple frequency bands and noisy sensing.
These factors introduce two difficulties in the problem: 1) due to the potentially large number of spectrum bands that need to be measured, the SUs would incur a 
significant cost to sense each spectrum band independently; in order to reduce this cost, we employ \emph{compressive spectrum sensing}, where
each SU collects a compressed measurement of the spectrum and transmits it to the CC, thus incurring only a fraction of the cost; this technique, in turns, complicates the design of the estimator and controller at the CC.
2) Due to the noise in the spectrum measurements, there is always some residual uncertainty in the current estimate, thus zero-interference operation is not possible (unless the SUs remain idle all the time); additionally, the accuracy of the spectrum estimate will depend on the number of measurements received at the CC, so that, the more the measurements, the better the estimate.
This factor introduces a requirement that $B\geq 1$, \emph{i.e.}, multiple control channels should be employed by the SUs to report their measurements,
as opposed to the noiseless case, where one measurement suffices, and thus $B=1$.
In Secs. \ref{scheduling} and \ref{sensing}, we introduce the models of spectrum scheduling and (compressed) spectrum sensing,  respectively,
and, in Sec. \ref{dynamics}, we characterize the dynamics of the system.

We consider a licensed spectrum composed of $F$ frequency bands, represented in Fig. \ref{TOYEX}.
Let $\mathbf b_k=(\mathbf b_{k,1},\mathbf b_{k,2},\dots,\mathbf b_{k,F})^T$ be the $F$-dimensional
spectrum occupancy (column) vector at time $k$, where $^T$ denotes matrix transpose,
and $\mathbf b_{k,i}\in\{0,1\}$ is the occupancy state of the $i$th band.

%The PUs occupy the spectrum according to a stochastic process, so that $\mathbf b_{k}$ is time-varying.

The system is time-slotted with slot duration $1$ and operates in two phases~\cite{MicheISIT}:
a \emph{sensing phase}, of duration $d$,
during which the SUs collect compressed distributed measurements of the spectrum occupancy state
and report them to a CC (\emph{e.g.}, a base station) (Sec. \ref{sensing}); followed by 
a \emph{scheduling phase}, of duration $1-d$,
where the SUs access the spectrum based on the  scheduling decision of the CC (Sec. \ref{scheduling}).
%{-3mm}
\vspace{-5mm}
\subsection{Spectrum Scheduling}
 \label{scheduling}
 In the spectrum scheduling phase,
  the SUs opportunistically access the spectrum based on the traffic vector
 decision $\mathbf r_k$ broadcasted by the CC,
 where $\mathbf r_k{=}(\mathbf r_{k,1},\mathbf r_{k,2},\dots,\mathbf r_{k,F})^T{\in}[0,1]^F$,
 and $\mathbf r_{k,i}$ is the average SU traffic in the $i$th spectrum band at time $k$.
 In each spectrum band, the dynamics of the PU system evolve as described in Sec. \ref{motivation}.

We define the aggregate expected instantaneous
throughput for the SU and PU systems, respectively,
given $\mathbf b_{k}$ and $\mathbf r_k$, as
 \begin{align}\label{rewX}
&T_{X}(\mathbf b_{k},\mathbf r_k)=\sum_{i=1}^FP_{succ}^{(X)}\left(\mathbf b_{k,i},\mathbf r_{k,i}\right),\ X\in\{S,P\}.
  \end{align}
 %{-5mm}
\vspace{-10mm}
\subsection{Spectrum Sensing}
 \label{sensing}
 At the beginning of slot $k$, the spectrum occupancy
 $\mathbf b_k$ is inferred by collecting
noisy
compressed spectrum measurements by the SUs,\footnote{We assume 
that the measurements are collected by the SUs. However, the analysis can be extended to the case where the 
sensors collecting the measurements and the SUs performing spectrum access do not coincide.}
according to the observation model (for SU $j$)
\begin{align}
\label{measmodel}
\mathbf y_{k,j}
=\mathbf a_{k,j}^T\mathbf b_{k}+\mathbf n_{k,j},\ \forall j=1,2,\dots,N_S,
\end{align}
%%{-6mm}\\
where $\mathbf n_{k,j}{\sim}\mathcal N(0,\sigma_Z^2)$\footnote{For simplicity and without loss of generality, we consider real-valued quantities. The following framework and analysis can be extended to complex-valued ones.}
is Gaussian noise, i.i.d. over time and across SUs,
 $\mathbf a_{k,j}^T$ is the measurement vector, and
 the superscript "T" denotes the matrix or vector transpose. 
 Eq.~(\ref{measmodel}) is the result of filtering over the spectrum band, so that $\mathbf a_{k,j}$ denotes the filtering 
coefficient vector, which includes also the signal attenuation between the PU and the SU.
We assume that $\mathbf a_{k,j}{\sim}\mathcal N(\mathbf 0,\sigma_A^2\mathbf I_F)$,
where $\mathbf I_n$ is the $n\times n$ identity matrix,
and is known to the CC.
\vspace{-5mm}
\begin{remark}
Note that each SU can, in principle, estimate the spectrum occupancy state $\mathbf b_{k}$ based only on local measurements $\mathbf y_{k,j}$. 
However, if $F$ is large, or the measurement is very noisy ($\sigma_Z^2/\sigma_A^2\gg 1$), the estimation accuracy may be very poor.
In contrast, by collecting measurements from a large number of SUs, the CC can estimate $\mathbf b_k$ more accurately.
\end{remark}

The SUs share $B$ orthogonal control channels to report their measurements, resulting in packet losses
if more than one SU transmit on the same channel.
The \emph{SU sensing traffic} in each control channel is $\psi_k$ in slot $k$, whose value
is broadcasted by the CC to the SUs at the beginning of slot $k$,
so that 
the SUs activate with common probability $\alpha_k=\psi_kB/N_S$,
and transmit their measurement in one of the $B$ channels available,
 incurring the sensing-transmission cost $c_{S}$. No cost is incurred by staying inactive. 

We denote the set of SUs that activate to sense and report their measurement
as $\mathcal A_k$ with cardinality $A_k$, and the set of SUs
that report successfully their measurement to the CC as $\mathcal M_k$
with cardinality $M_k$.
We define the probability that  $M_k=m$ measurements are successfully received at the CC, given that $A_k=a$ SUs activate,
as $p_{M|A}(m|a)=\mathbb P(M_k=m|A_k=a)$.
 Moreover, we define the probability that  $M_k=m$ measurements are successfully received at the CC, given the sensing traffic $\psi_k$,
 as $p_{M}(m|\psi_k)=\mathbb E_{A_k}[p_{M|A}(m|A_k)|M_k=m,\psi_k]$, by taking the expectation with respect to $A_k\sim\mathcal B(N_S,\alpha_k)$.
Assuming a collision model for the $B$ control channels and $N_S\to\infty$, the number of measurements received at the CC,
  $M_k$, has binomial distribution with
$B$ trials and success probability $\psi_k e^{-\psi_k}$
 \cite{MicheTSP1}, \emph{i.e.},
\begin{align*}
p_{M}(m|\psi_k){=}
\mathbb P(M_k{=}m|\psi_k){=}\!
\left(\begin{array}{c}\!\!\!B\!\!\!\\\!\!\!m\!\!\!\end{array}\right)\!\!\left(\psi_k e^{-\psi_k}\right)^m\!\!\left(1{-}\psi_k e^{-\psi_k}\right)^{B-m}\!\!\!,
\end{align*}
so that the expected number of measurements received is $\mathbb E[M_k|\psi_k]=B\psi_k e^{-\psi_k}$.
In the following treatment, we use this approximation, although the model can be extended to finite $N_S$ and more general channel models,
by defining $p_{M}(m|\psi_k)$ accordingly.

 Let $\mathbf y_k\in\mathbb R^{M_k}$ be the vector of \emph{compressed measurements} collected in slot $k$.
 From (\ref{measmodel}), 
  %%{-1mm}
 \begin{align}\label{yt}
\mathbf y_k=\mathbf A_k^T\mathbf b_{k}+\mathbf n_{k},
\end{align}
 %%{-5mm}\\
where $\mathbf A_k=[\mathbf a_{k,j}]_{j\in\mathcal M_k}$ is the measurement matrix, known to the CC,
and  $\mathbf n_k=[z_{k,j}]_{j\in\mathcal M_k}$ is the noise column vector.
Note that the size of $\mathbf y_k$, $M_k$, is random, due to the probabilistic activation decision of each SU and packet losses resulting from the shared wireless control channels.
\vspace{-5mm}
\subsection{System dynamics}
\label{dynamics}

  \begin{figure}[t]
\centering
\includegraphics[width = .5\linewidth,trim = 0mm 1mm 0mm 1mm,clip=false]{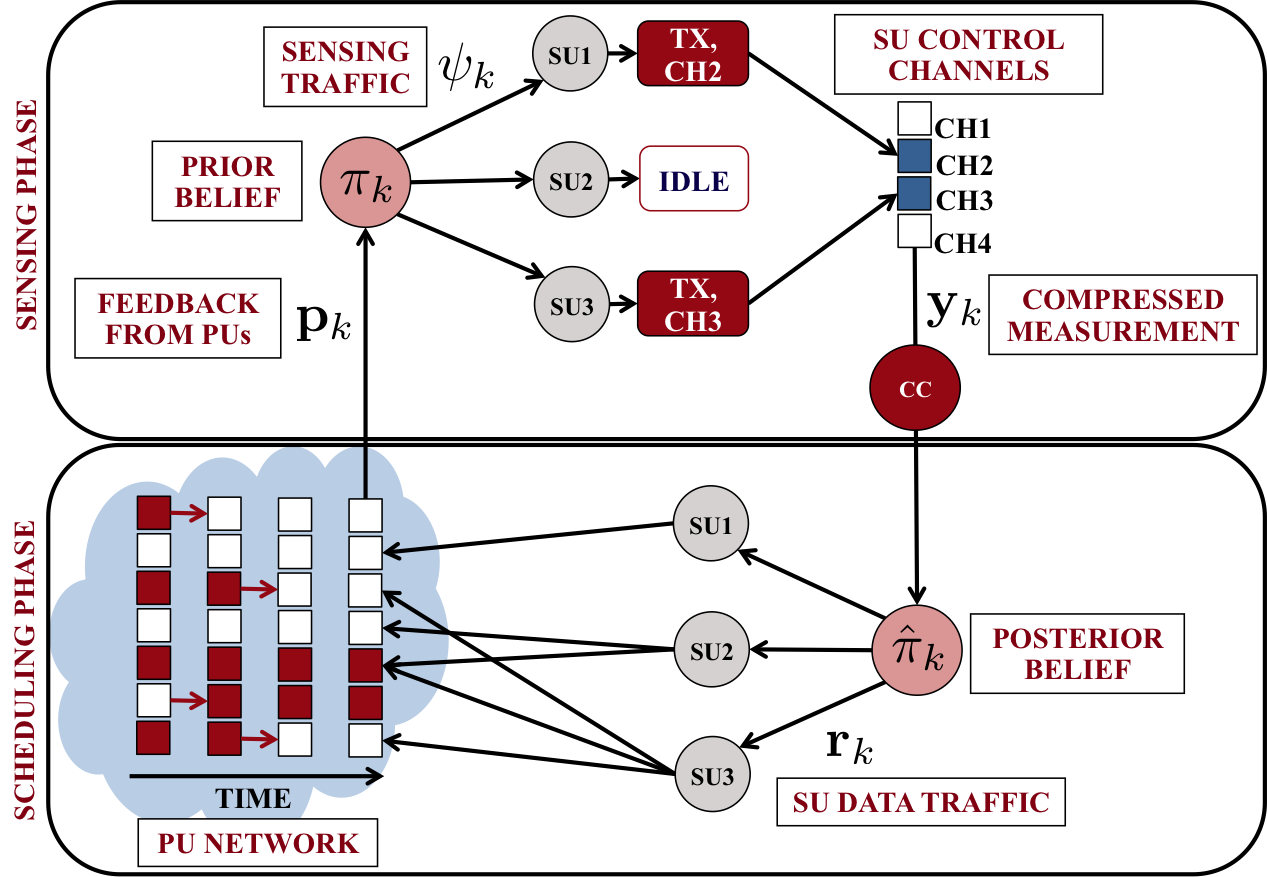}
%{-5mm}
\vspace{-5mm}
\caption{Block diagram of the system dynamics.}
\label{blockdiag}
\vspace{-9mm}
\end{figure}

The dynamics of the system in each slot can be summarized as follows (see also Fig. \ref{blockdiag}):

\noindent \textbf{1}) \emph{Sensing phase}: At the beginning of slot $k$, the sensing traffic $\psi_k$ is selected by the CC, and broadcasted to the SUs;
each SU collects a compressed measurement  with probability $\alpha_k=\psi_kB/N_S$ and transmits it independently in one of the $B$ control channels;

\noindent \textbf{2})  \emph{Measurement collection}: The measurement vector $\mathbf y_k\sim\mathcal N(\mathbf A_k^T\mathbf b_k,\sigma_Z^2\mathbf I_{R_k})$
is collected at the CC, where $M_k\sim P_M(M_k|\psi_k)$;

\noindent \textbf{3})  \emph{Scheduling phase}: the traffic vector $\mathbf r_k\in [0,1]^F$ is chosen by the CC and broadcasted to the SUs;
each SU transmits its own data with probability $\mathbf q_{k,i}=\mathbf r_{k,i}/N_S$ in the $i$th band;

\noindent \textbf{4}) \emph{State dynamics}: state transitions with 
 probability $P_{B}(\mathbf b_{k+1,i}|\mathbf b_{k,i},\mathbf r_{k,i})$ in the $i$th spectrum band.

We denote the \emph{prior belief}
that $\mathbf b_k{=}\mathbf b$, based on the history collected up to time $k$,
denoted as $\mathcal H_k$,
and before the sensing phase, as $\pi_k(\mathbf b){=}\mathbb P(\mathbf b_k{=}\mathbf b|\mathcal H_k)$.
Similarly, 
we denote the \emph{posterior belief}
that $\mathbf b_k{=}\mathbf b$, given $(\mathcal H_k,\mathbf y_k,\mathbf A_k)$,
 as $\hat\pi_k(\mathbf b){=}\mathbb P(\mathbf b_k{=}\mathbf b|\mathcal H_k,\mathbf y_k,\mathbf A_k)$.
Using (\ref{yt}),  we have that
\begin{align}
\label{posterior}
\hat\pi_k(\mathbf b)
\propto\pi_k(\mathbf b)
\exp\left\{-\frac{1}{2\sigma_Z^2
}\left\Vert\mathbf y_k-\mathbf A_k^T\mathbf b\right\Vert_F^2
\right\},
\end{align}
where $\propto$ denotes proportionality up to a normalization factor,
so that we can write $\hat\pi_k=\hat\Pi(\pi_k,\mathbf A_k,\mathbf y_k)$,
for a proper function $\hat\Pi(\cdot)$.

The CC, at the end of the slot, may overhear the PU acknowledgments of correct (ACK) or incorrect (NACK) reception of the packets, fed back by the PU receivers on
 each channel, denoted as $\mathbf p_{k,i}\in\{\text{ACK},\text{NACK},\emptyset\}$, where 
 $\mathbf p_{k,i}=\emptyset$ if either an erasure occurs (the ACK/NACK message cannot be detected by the CC)
 or the $i$th band was idle, so that no feedback information is reported. We denote the erasure probability as $\epsilon\in [0,1]$,
and the feedback vector collected at the CC at the end of slot $k$ as $\mathbf p_{k}$.
Given $\mathbf b_{k,i}$ and $\mathbf r_{k,i}$, the probability mass function (pmf) of $\mathbf p_{k,i}$,
denoted as
$P_{P}(p|b,r)\triangleq
\mathbb P\left(\mathbf p_{k,i}=p|\mathbf b_{k,i}=b,\mathbf r_{k,i}=r\right)$,
 is given by
\begin{align}
\label{ppp1}
&
P_{P}(\emptyset|b,r)=1-b+b\epsilon,
\ 
P_{P}(\text{ACK}|b,r)=(1-\epsilon)P_{succ}^{(P)}(b,r),
\\&
P_{P}(\text{NACK}|b,r)=(1-\epsilon)(b-P_{succ}^{(P)}(b,r)),
\label{pppend}
\end{align}
where $\chi(\cdot)$ is the indicator function.
Therefore, the pmf of $\mathbf p_{k}$
 given $\mathbf b_{k}$ and $\mathbf r_{k}$ is given by
 \begin{align}
 \label{Distrp}
 \mathbb P\left(\mathbf p_{k}=\mathbf p|\mathbf b_{k,i}=\mathbf b,\mathbf r_{k,i}=\mathbf r\right)
 =\prod_i
 P_{P}(\mathbf p_{i}|\mathbf b_{i},\mathbf r_i).
 \end{align}
 
Given $\mathbf r_k$ and $\mathbf p_{k}$,
 the CC updates the next prior belief as
\begin{align}\label{PI}
&\!\!\!\pi_{k+1}(\mathbf b){=}\!\sum_{\tilde{\mathbf b}}\hat\pi_k(\tilde{\mathbf b})
\prod_{i=1}^F
P_{B|P}(\mathbf b_{i}|\tilde {\mathbf b}_{i}, \mathbf r_{k,i},\mathbf p_{k,i}),
\end{align}
where
$P_{B|P}(b^\prime|b,r,p){\triangleq}\mathbb P(\mathbf b_{k+1,i}{=}b^\prime|\mathbf b_{k,i}{=}b,\mathbf r_{k,i}{=}r,\mathbf p_{k,i}{=}p)$,
given by
\begin{align}
\label{pp1}
&P_{B|P}(1|1,r,\text{ACK})=\theta+(1-\theta)\zeta,\qquad
P_{B|P}(1|1,r,\text{NACK})=1,
\\
\label{pp3}
&P_{B|P}(1|b,r,\emptyset)=
P_{succ}^{(P)}(b,r)\epsilon(\theta+(1-\theta)\zeta)
%\\\nonumber&\qquad
+(b-P_{succ}^{(P)}(b,r))\epsilon
+(1-b)\zeta,
\end{align}
and $P_{B|P}(0|b,r,p){=}1{-}P_{B|P}(1|b,r,p)$.
Note that ACK/NACK reception ($\mathbf p_{k,i}{\neq}\emptyset$) implies $\mathbf b_{k,i}{=}1$.
We can thus write
$\pi_{k+1}{=}\Pi(\hat\pi_k,\mathbf r_k,\mathbf p_k)$,
for a proper function $\Pi(\cdot)$.
\vspace{-5mm}
\section{Policy definition and Optimization Problem}
\label{optimiz}
In this section, we present the spectrum sensing
and scheduling policies  (Sec. \ref{optimizA}), and we introduce the performance metrics and the optimization problem (Sec. \ref{optimizP}).
Complexity reduction techniques will be 
 carried out in the following Secs.~\ref{optimization} and \ref{complereduction}.
\vspace{-5mm}
\subsection{Spectrum Sensing and Scheduling policies}
\label{optimizA}
In the sensing phase, given $\mathcal H_k$, the CC choses $\psi_k$
 according 
to a sensing policy $\psi_k{=}\psi(\mathcal H_k)$.
In the scheduling phase, given  $\hat{\mathcal H}_k{=}(\mathcal H_k,\mathbf y_k,\mathbf A_k)$,
 the CC selects $\mathbf r_k$
 according to a scheduling policy
$\mathbf r_k{=}\mathbf r(\hat{\mathcal H}_k)$.
We denote the \emph{joint sensing-scheduling} policy as $(\psi,\mathbf r)$.
\vspace{-5mm}
\subsection{Performance metrics and Optimization problem}
\label{optimizP}
We define the average long-term sensing and data transmission cost of the SU network as
\begin{align}
\label{cSensing}
&\bar C_{sensing}(\psi,\mathbf r){\triangleq}\lim_{D\to\infty}\frac{1}{D}\mathbb E\left[\left.\sum_{k=0}^{D-1}
\psi_kBc_S\right|\pi_0\right],
\\
\label{csched}
&\bar C_{sched}(\psi,\mathbf r){\triangleq}\lim_{D\to\infty}\frac{1}{D}\mathbb E\left[\left.\sum_{k=0}^{D-1}
\sum_{i=1}^F\mathbf r_{k,i}c_{TX}\right|\pi_0\right],
\end{align}
respectively,
where $\pi_0$ is the initial prior belief at the beginning of slot $0$.
We define the total cost of sensing and data transmission as 
\begin{align}
\bar C(\psi,\mathbf r)\triangleq
\bar C_{sensing}(\psi,\mathbf r)
+\bar C_{sched}(\psi,\mathbf r).
\end{align}
Finally, we define the average SU/PU throughputs as
\begin{align}
\label{VX}
\bar T_{X}(\psi,\mathbf r)\triangleq\lim_{D\to\infty}\frac{1}{D}\mathbb E\left[\left.\sum_{k=0}^{D-1}
T_{X}(\mathbf b_{k},\mathbf r_k)
\right|\pi_0\right],\ X\in\{S,P\},
\end{align}
where the expectation is with respect to the
realization of $\{\mathbf b_k,\mathbf A_k,\mathbf y_k,\mathbf r_k,\mathbf p_{k}\}$, induced by $(\psi,\mathbf r)$.
The goal is to determine  the joint sensing-scheduling policy $(\psi^*,\mathbf r^*)$ such that
\begin{align}
\label{optprob}
(\psi^*,\mathbf r^*)&=\arg\max_{(\psi,\mathbf r)}
\bar T_{S}(\psi,\mathbf r)
%\nonumber\\
\text{ s.t.}\ \
% &
\bar C(\psi,\mathbf r)\leq \bar C^{\max},\ 
%\nonumber\\&
\bar T_{P}(\psi,\mathbf r)\geq \bar T_P^{\min},
\end{align}
where $\bar C^{\max}$ is the maximum cost of sensing and data transmission, and
$\bar T_P^{\min}$ is the minimum PU throughput requirement. 
\begin{comment}
We have the following result, establishing the feasiliblity of (\ref{optprob}).
\begin{propos}
The optimization problem (\ref{optprob}) is feasible if and only if
\begin{align}
&\bar C^{\max}\geq 0,
\ 
\bar T_P^{\min}\leq F\pi_P(1)\triangleq \bar T_P^{\max}.
\end{align}
\end{propos}
\begin{proof}
The proof directly follows from the fact that the sensing-scheduling  policy which
assigns $\psi_k=0$, $\mathbf r_k=\mathbf 0,\forall k$ incurs minimum
cost $\bar C(\psi,\mathbf r)=0$ and maximum PU throughput 
$\bar T_{P}(\psi,\mathbf r)=\bar T_P^{\max}$, and is thus a feasible policy, for any value of 
$\bar C^{\max}$ and $\bar T_P^{\min}$ satisfying the proposition.
\end{proof}
\end{comment}
Alternatively, we consider the Lagrangian formulation % of the optimization problem (\ref{optprob}),
\begin{align}
\label{optproblag}
(\psi^*,\mathbf r^*)=\arg\max_{(\psi,\mathbf r)}&\ 
\xi\bar T_{S}(\psi,\mathbf r)+(1-\xi)\bar T_{P}(\psi,\mathbf r)
%\nn\\&
-\lambda\bar C(\psi,\mathbf r),
\end{align}
where the parameters $\lambda{\geq}0$ and $\xi{\in}(0,1)$ capture the desired trade-off
between achieving high PU/SU throughputs and incurring low cost for data transmission and acquisition of state information at the CC.
We have the following theorem.
%{-2mm}
\vspace{-5mm}
\begin{thm}
\label{suffstat}
The prior belief $\pi_k$ is a sufficient statistic to choose the sensing action $\psi_k$ in slot $k$.
The posterior belief $\hat\pi_k$ is a sufficient statistic to choose the traffic $\mathbf r_k$ in slot $k$.
\end{thm}
\vspace{-3mm}
\begin{proof}
See Appendix A.
\end{proof}

We can thus restrict the design to \emph{stationary} policies of the form $\psi_k=\psi(\pi_k)$
and $\mathbf r_k=\mathbf r(\hat\pi_k)$ which depend solely on the respective sufficient statistic,
so that we can rewrite
\begin{align*}
&\bar C_{sensing}(\psi,\mathbf r)
=
Bc_S\lim_{D\to\infty}\frac{1}{D}
\mathbb E\left[\left.\sum_{k=0}^{D-1}
\psi(\pi_k)\right|\pi_0\right],
\\&
\bar C_{sched}(\psi,\mathbf r)
=
c_{TX}\lim_{D\to\infty}\frac{1}{D}\mathbb E\left[\left.\sum_{k=0}^{D-1}
\sum_{i=1}^F\mathbf r_i(\hat\pi_k)\right|\pi_0\right],
\\&
\bar T_{X}(\psi,\mathbf r)
=
\lim_{D\to\infty}\frac{1}{D}\mathbb E\left[\left.\sum_{k=0}^{D-1}
\sum_{\mathbf b}\hat\pi_k(\mathbf b)
T_{X}(\mathbf b,\mathbf r(\hat\pi_k))
\right|\pi_0\right],
\end{align*}
where the expectation is taken with respect to the sequence $\{\pi_k,\hat\pi_k,k\geq 0\}$, induced by $(\psi,\mathbf r)$.
\vspace{-5mm}
 \section{Optimization techniques}
 \label{optimization}
 In this section, we develop optimization techniques to solve 
 the optimization problem (\ref{optprob}) with lower complexity. In particular, in Sec. \ref{optimalDP}, we first
 introduce the optimal DP algorithm, which exploits Theorem \ref{suffstat} to  decouple the optimization of sensing and scheduling, and discuss its enormous complexity.
 Then, in Sec. \ref{parmyop}, we present our proposed partially-myopic scheduling scheme, which enables complexity reduction in the DP optimization.
 Due to the POMDP formulation, in Sec. \ref{complereduction} we will resort to
 belief approximation based on KLD minimization, which enables the use of sparse recovery techniques to estimate the spectrum occupancy.
 \vspace{-5mm}
 \subsection{Optimal DP algorithm: decoupling the optimization of sensing and scheduling}
 \label{optimalDP}
 The optimal solution of (\ref{optproblag}) can by found via DP.
 In particular, we can exploit Theorem \ref{suffstat} to decouple the DP algorithm into two sub-stages,
 which exploit the different sufficient statistic used in the spectrum sensing and scheduling phases, respectively.
\begin{algo}[Optimal sensing-scheduling DP]
\label{aDP}%$ $ \\

\noindent {\bf 1)} Initialize $V^{[0]}(\hat \pi)=0$, $\forall \hat\pi$; $l=1$;

\noindent {\bf 2)} {\bf Scheduling optimization stage}: in the $l$th iteration, determine, $\forall \hat\pi$,
\begin{align}
&\hat V^{[l]}(\hat \pi)=
\max_{\mathbf r\in[0,1]^F}\sum_{\mathbf b}\hat\pi(\mathbf b)
\left\{
\xi T_{S}(\mathbf b,\mathbf r)
+(1-\xi)T_{P}(\mathbf b,\mathbf r)
\vphantom{\left[\left.V^{[l-1]}(\Pi(\hat\pi,\mathbf r,\mathbf p))\right|\mathbf b,\mathbf r\right]}
\right.
\nonumber\\&
\left.
\qquad-\lambda c_{TX}\mathbf 1^T\mathbf r 
+\mathbb E\left[\left.V^{[l-1]}\left(\Pi\left(\hat\pi,\mathbf r,\mathbf p\right)\right)\right|\mathbf b,\mathbf r\right]
\right\},
\label{Dpsched}
\end{align}
where the expectation is with respect to the realization of $\mathbf p$, conditioned on $\mathbf b$ and $\mathbf r$;
the maximizer is the optimal SU data traffic in the $l$th stage, $\mathbf r^{[l]}(\hat\pi)$;

\noindent {\bf 3)} {\bf Sensing evaluation stage}: in the $l$th iteration, determine, $\forall\pi$, $\forall m\in\{0,1,\dots,B\}$,
\begin{align*}
V_m^{[l]}(\pi)=\sum_{\mathbf b}\pi(\mathbf b)\mathbb E\left[\left.\hat V^{[l]}\left(\hat\Pi\left(\pi,\mathbf A^{(m)},\mathbf y^{(m)}\right)\right)\right|\mathbf b,m\right],
\end{align*}
where the expectation is with respect to the realization of 
the measurement matrix $\mathbf A^{(m)}\in\mathbb R^{F\times m}$
and measurement vector $\mathbf y^{(m)}|\mathbf b\sim\mathcal{N}(\mathbf A^{(m),T}\mathbf b,\sigma_w^2\mathbf I_m)$,
 conditioned on $\mathbf b$ and the number of measurements received, $m$;
 
\noindent {\bf 4)} {\bf Sensing optimization stage}: in the $l$th iteration, determine, $\forall\pi$,
\begin{align*}
V^{[l]}(\pi){=}\!\!\max_{\psi\in[0,1]}\!
\sum_{\mathbf b}\pi(\mathbf b)\!\!
\left\{
\!{-}\lambda \psi Bc_S
{+}\!\!\sum_{m=0}^{B}p_{M}(m|\psi)V_m^{[l]}(\pi)\!\right\}\!,
\end{align*}
 the maximizer is the optimal sensing traffic in the $l$th stage, $\psi^{[l]}(\pi)$;
 
\noindent {\bf 5)} repeat from step 2) with $l:=l+1$ until convergence; return
policy $(\mathbf r^{[l]},\psi^{[l]})$.
\end{algo}
\vspace{-5mm}
\begin{remark}
The term $V_m^{[l]}(\pi)$ in step 3) represents an evaluation of the cost-to-go function when $m$ measurements are collected at the CC,
 and is independent of the scheme employed by the SUs to report their measurements. On the other hand, the 
 term $V^{[l]}(\pi)$ in step 4) evaluates the cost-to-go function under the specific reporting scheme, as described in Sec. \ref{sensing}.
\end{remark}

Importantly,  Theorem \ref{suffstat}
allows us to relax the \emph{joint} optimization of the sensing and scheduling actions and, instead,
decouple it into two sub-stages: the first one, \emph{scheduling optimization stage},
uses only the posterior belief information to determine the optimal SU data traffic; the second one, \emph{sensing optimization stage},
uses only the prior belief information to determine the optimal SU sensing traffic. The proposed algorithm, thus, effectively captures the sequential structure of the decision making process, \emph{i.e.},
the prior belief drives the sensing traffic, which, in turn, determines the posterior belief, based on which the SU data traffic is scheduled, and so on.

Despite the complexity reduction obtained by decoupling the DP optimization into sub-stages, the DP algorithm has enormous complexity, due to the POMDP formulation and the huge action space.
 In particular, the prior and posterior beliefs $\pi$ and $\hat\pi$
are defined over a $2^F$ dimensional space of all possible realizations of the PU spectrum occupancy state, leading to the curse of dimensionality.
In Sec. \ref{complereduction},
the POMDP formulation is relaxed by projecting the prior and posterior beliefs on a lower dimensional manifold,
thus leading to a compact belief representation.

Additionally, the SU traffic $\mathbf r$ is defined over the set $[0,1]^F$, leading to 
huge complexity in the scheduling optimization stage due to the huge  action space. 
In Sec. \ref{parmyop}, we propose a \emph{partially myopic} scheduling to relax this dimensionality issue.
\vspace{-5mm}
\subsection{Partially Myopic scheduling scheme}
\label{parmyop}
Let $\Lambda_k{=}\sum_{i}\mathbf r_{k,i}$ be the total SU traffic budget in the scheduling phase in slot $k$.
Then, we can decouple the scheduling policy $\mathbf r(\hat\pi)$ into the following sub-policies:
a policy $\Lambda(\hat\pi){\in}[0,F]$ which decides on the total traffic budget allocated as a function of $\hat\pi$,
and a policy $\mathbf z(\hat\pi){\in}\mathcal Z$, which assigns the total budget to the
different spectrum bands, where
$\mathcal Z{\equiv}\{\mathbf z:\sum\mathbf z_i=1,\mathbf z\geq 0\}$.
We can thus rewrite the one-to-one mapping between $\mathbf r$ and $(\Lambda,\mathbf z)$
\begin{align}
\mathbf r(\hat\pi)=\Lambda(\hat\pi)\mathbf z(\hat\pi).
\end{align}

Then, step 2) in the DP Algorithm \ref{aDP} can be replaced with
\begin{align}
\label{asdasd}
&\hat V^{[l]}(\hat \pi)=
\max_{\Lambda,\mathbf z}\sum_{\mathbf b}\hat\pi(\mathbf b)
\left\{
\xi T_{S}(\mathbf b,\Lambda\mathbf z)
+(1-\xi)T_{P}(\mathbf b,\Lambda\mathbf z)
\vphantom{\left[\left.V^{[l-1]}(\Pi(\hat\pi,\Lambda\mathbf z,\mathbf p))\right|\mathbf b,\Lambda\mathbf z\right]}
\right.\nonumber\\&\left.
-\lambda c_{TX}\Lambda
+\mathbb E\left[\left.V^{[l-1]}(\Pi(\hat\pi,\Lambda\mathbf z,\mathbf p))\right|\mathbf b,\Lambda\mathbf z\right]
\right\},
\end{align}
thus yielding the optimal $\Lambda^{[l]}(\hat\pi)$ and $\mathbf z^{[l]}(\hat\pi)$.
Note that the optimization over $\Lambda{\in}[0,F]$ can be carried out with complexity linear in $F$, since the total traffic budget $\Lambda$ is a scalar quantity
taking value in the closed set $[0,F]$.
On the other hand, the optimization over $\mathbf z$ has high complexity since $\mathbf z{\in}\mathcal Z$,
and the action space $\mathcal Z$ grows exponentially with the number of frequency bands $F$.
In order to reduce the complexity,
using a similar approach as in \cite{MicheEH},
 we use a myopic approach to
 approximate $\mathbf z(\hat\pi,\Lambda)$ for a given total budget $\Lambda$,
 namely,
 \begin{align}
\mathbf z(\hat\pi,\Lambda)&=
\arg\max_{\mathbf z\in\mathcal Z}
\sum_{\mathbf b}\hat\pi(\mathbf b)
\left[\xi T_{S}(\mathbf b,\Lambda\mathbf z)
%\right.\nonumber\\&\left.
+(1-\xi)T_{P}(\mathbf b,\Lambda\mathbf z)
-\lambda c_{TX}\Lambda\right],
 \end{align} 
 which corresponds to the instantaneous cost in the DP stage  (\ref{Dpsched}),
 without the cost-to-go term $\mathbb E\left[V^{[l-1]}\right]$.
Using (\ref{PsuccPU}), (\ref{PsuccSU}) and (\ref{rewX}),
 we can rewrite this optimization problem as 
  \begin{align}\label{optz}
\mathbf z(\hat{\boldsymbol\beta},\Lambda)&{=}
\arg\max_{\mathbf z\geq 0}
\sum_{i=1}^F\left[(1{-}\hat{\boldsymbol{\beta}}_{i})\Lambda\mathbf  z_{i}+\frac{1-\xi}{\xi}\frac{1{-}\rho_P}{1{-}\rho_S}\hat{\boldsymbol{\beta}}_{i} \right]e^{-\Lambda \mathbf z_{i}}
%\nonumber\\&\qquad
\text{ s.t.}\sum_i\mathbf z_{i}=1,
 \end{align}
 where we have defined the expected posterior occupancy vector
 \begin{align}
 \hat{\boldsymbol\beta}=\sum_{\mathbf b}\mathbf b\hat\pi(\mathbf b),
 \end{align}
 and we have expressed $\mathbf z(\hat{\boldsymbol\beta},\Lambda)$ as a function of 
 $\hat{\boldsymbol\beta}$ only, rather than of the posterior belief $\hat\pi$.
  
 Additionally, we can further bound the feasible values of the total traffic budget $\Lambda$ as follows.
Let $\mathbf r_{\max}(\hat{\boldsymbol\beta})$ be the solution of the unconstrained optimization problem
  \begin{align}
  \label{asdfasddfdf}
&\mathbf  r_{\max}(\hat{\boldsymbol\beta}){=}
 \arg\max_{\mathbf r\geq 0}
\sum_{\mathbf b}\hat\pi(\mathbf b)
\left[\xi T_{S}(\mathbf b,\mathbf r)+(1-\xi) T_{P}(\mathbf b,\mathbf r)\right]
\\&\nonumber
=
\arg\max_{\mathbf r\geq 0}
\sum_i\left[\xi(1{-}\hat{\boldsymbol\beta}_i)(1{-}\rho_S) \mathbf r_i+(1{-}\xi)\hat{\boldsymbol\beta}_i(1{-}\rho_P)\right]e^{-\mathbf r_i}
%\nonumber\\&\qquad
=\left[1-\frac{\hat{\boldsymbol\beta}}{1-\hat{\boldsymbol\beta}}
\frac{(1-\xi)(1-\rho_P)}{\xi(1-\rho_S)}
\right]^+,
 \end{align}
 where we have defined $[\cdot]^+=\max\{\cdot,0\}$ and component-wise operations.
Note that $\mathbf  r_{\max}(\hat{\boldsymbol\beta})$ is the value of the SU traffic which maximizes 
the trade-off between the instantaneous PU and SU throughputs, as a function of the expected occupancy $\hat{\boldsymbol\beta}$.
If the SU traffic in the $i$th spectrum band is such that $\mathbf r_{k,i}>\mathbf  r_{\max,i}(\hat{\boldsymbol{\beta}}_{k})$,
then the following undesirable outcomes occur:
a smaller trade-off between PU and SU throughputs is achieved, since
$\mathbf  r_{\max,i}(\hat{\boldsymbol{\beta}}_{k})$ optimizes such trade-off (see (\ref{asdfasddfdf}));
a larger transmission cost is incurred by the SUs in the scheduling phase;
collisions to the PU operating in the $i$th spectrum band are more likely to occur,
so that the $i$th spectrum band is more likely to be occupied in the next slot, due to the retransmission mechanism.
Therefore, we restrict $\mathbf r(\hat\pi_k)$ to take values
$\mathbf 0\leq\mathbf r(\hat\pi_k)\leq  \mathbf r_{\max}(\hat{\boldsymbol\beta}_k)$,
so that
\begin{align}
\Lambda(\hat\pi)\leq\sum_i  \mathbf r_{\max,i}(\hat{\boldsymbol\beta})\triangleq\Lambda_{\max}(\hat{\boldsymbol\beta}),
\end{align}
and,
for a given $\Lambda\in[0,\Lambda_{\max}(\hat{\boldsymbol\beta})]$,
 $\mathbf z\leq \frac{\mathbf r_{\max}(\hat{\boldsymbol\beta})}{\Lambda}$.
Hence,
 (\ref{optz}) is equivalent to
  \begin{align}\label{optz2}
\!\!\!\!\!\mathbf z(\hat{\boldsymbol\beta},\Lambda)&{=}
\arg\max_{\mathbf z}
\sum_{i=1}^F\left[(1{-}\hat{\boldsymbol{\beta}}_{i})\Lambda\mathbf z_{i}+\frac{1-\xi}{\xi}\frac{1{-}\rho_P}{1{-}\rho_S}\hat{\boldsymbol{\beta}}_{i} \right]e^{-\Lambda \mathbf z_{i}}
%\nonumber\\&
\ \text{ s.t.}\sum_i\mathbf z_{i}{=}1,\mathbf 0{\leq}\mathbf z{\leq}\frac{\mathbf r_{\max}(\hat{\boldsymbol\beta})}{\Lambda}.
 \end{align}
 The partially myopic scheme and the optimization problem (\ref{optz2})
 have the following properties.
 \vspace{-5mm}
 \begin{thm}
 \label{thm2}
 \noindent {\bf 1)} The optimization problem (\ref{optz2}) is concave; 
 
 \noindent {\bf 2)}
If  $\hat{\boldsymbol\beta}_i\geq\frac{\xi(1-\rho_S)}{(1-\xi)(1-\rho_P)+\xi(1-\rho_S)}$ for some $i$, then
$\mathbf r_i(\hat{\boldsymbol\beta},\Lambda)=0$.

 \noindent {\bf 3)} The SU traffic
 $\mathbf r(\hat{\boldsymbol\beta},\Lambda)=\Lambda\mathbf z(\hat{\boldsymbol\beta},\Lambda)$ is a non-decreasing function of $\Lambda$ (component-wise);

\noindent {\bf 4)}
for a given $\Lambda\in[0,\Lambda_{\max}(\hat{\boldsymbol\beta})]$,
  if $\hat{\boldsymbol\beta}_i>\hat{\boldsymbol\beta}_j$ for some $i\neq j$, then
$
\mathbf r_i(\hat{\boldsymbol\beta},\Lambda)\leq \mathbf r_j(\hat{\boldsymbol\beta},\Lambda)$;
  \end{thm}
  \begin{proof}
  See Appendix B.
  \end{proof}

  Property {\bf 3)} states that, when the budget $\Lambda$ increases, the traffic scheduled in each spectrum band does not decrease;
  Properties {\bf 2)} and {\bf 4)} state that more traffic is scheduled in those bands more likely to be idle,
and no traffic is scheduled in those bands likely to be occupied by a PU.
   All these properties are desirable, since they ensure that the SU traffic is scheduled only to those bands more likely to be idle,
   thus minimizing the interference to the PUs.
      The implication of Property  {\bf 1)} is that (\ref{optz2}) can be solved
 efficiently using standard convex optimization tools~\cite{Boyd}.

   While $\mathbf z(\hat{\boldsymbol\beta},\Lambda)$ is obtained myopically as the solution of problem (\ref{optz2}),
the total traffic budget $\Lambda(\hat\pi)$ is determined optimally as the solution of the DP stage
   \begin{align}
\label{asdasd2}
\hat V^{[l]}(\hat \pi)=&
\max_{\Lambda}\sum_{\mathbf b}\hat\pi(\mathbf b)
\left\{
\xi T_{S}(\mathbf b,\Lambda\mathbf z(\hat{\boldsymbol\beta},\Lambda))
%\vphantom{\left[\left.V^{[l-1]}(\Pi(\hat\pi,\Lambda\mathbf z(\hat{\boldsymbol\beta},\Lambda),\mathbf p))\right|\mathbf b,\Lambda\mathbf z(\hat{\boldsymbol\beta},\Lambda)\right]}
%\right.\nonumber\\&\left.
+(1-\xi)T_{P}(\mathbf b,\Lambda\mathbf z(\hat{\boldsymbol\beta},\Lambda))-\lambda c_{TX}\Lambda
\right.\\&\left.\nonumber
+\mathbb E\left[\left.V^{[l-1]}(\Pi(\hat\pi,\Lambda\mathbf z(\hat{\boldsymbol\beta},\Lambda),\mathbf p))\right|\mathbf b,\Lambda\mathbf z(\hat{\boldsymbol\beta},\Lambda)\right]
\right\},
\end{align}  
   which replaces step 2) in the DP Algorithm \ref{aDP} and can be solved with linear complexity, rather than exponential complexity
   as in the original DP step 2); hence the name \emph{partially myopic} scheduling scheme,
   obtained by combining an \emph{optimal DP} solution of the total traffic budget with a \emph{myopic} scheduling of the total traffic budget across
   spectrum bands.
\vspace{-5mm}
 \section{Complexity reduction}
 \label{complereduction}
 Although the dynamics of the spectrum bands evolve independently across frequency, the compressed spectrum measurements
(\ref{yt}) introduce frequency correlation, as is evident from the belief update (\ref{posterior}).
Therefore, in general, the information available at the CC is represented by a belief $\pi(\mathbf b)$,
which may not factorize across frequency bands, resulting in high dimensionality
and huge optimization and operational complexity of the system.

In this section, we propose a compact belief representation, which makes it possible to optimize and operate the system on a lower dimensional subspace.
In particular, in Sec. \ref{seccompact}, we will resort to a compact belief representation via KLD minimization.
Then, in Sec. \ref{MAP}, we will show how this compact representation can be exploited to design sparse recovery techniques to estimate the spectrum occupancy.
Finally, in Sec. \ref{txprob}, we discuss the computation of the transition probabilities in the compact belief representation, which are required
in the DP algorithm.
 \vspace{-5mm}
\subsection{Compact belief representation via KLD minimization}
\label{seccompact}

In order to reduce the high dimensionality entailed by the POMDP formulation, we propose a
 compact state space representation by projecting the belief onto  a low-dimensional manifold via KLD
 minimization.
We approximate the belief $\pi(\mathbf b)$ with the factorized model
\begin{align}
\label{approxbeilef}
{\pi}(\mathbf b)\simeq
\tilde{\pi}(\mathbf b)= \prod_i [\bar\beta^{(\phi(i))}]^{\mathbf b_i}[1-\bar\beta^{(\phi(i))}]^{1-\mathbf b_i},
\end{align}
where $\bar\beta^{(L)},\bar\beta^{(H)}\in[0,1]$ with $\bar\beta^{(L)}\leq\bar\beta^{(H)}$ are low (L) and high (H) probability levels, and $\phi:\{1,2,\dots,F\}\mapsto\{L,H\}$ is a function which maps the 
$i$th spectrum band to indices corresponding to one of the levels $\bar\beta^{(L)}$ or $\bar\beta^{(H)}$. Note that this approximation assumes that
the spectrum bands are statistically independent of each other, and that their probability of being occupied takes two possible values, $\bar\beta^{(L)}$ or $\bar\beta^{(H)}$.
We can alternatively interpret the bands with high probability of occupancy $\bar\beta^{(H)}$ as those detected to be occupied, so that
$P_{FA}=1-\bar\beta^{(H)}$ is the corresponding \emph{false-alarm} probability.
Similarly, the bands with low probability of occupancy $\bar\beta^{(L)}$ 
are those detected to be idle, so that
$P_{MD}=\bar\beta^{(L)}$ is the corresponding \emph{missed-detection} probability.

The approximate belief $\tilde{\pi}(\mathbf b)$
 is parameterized by  $(\bar\beta^{(L)},\bar\beta^{(H)},\phi(\cdot))$. We thus denote $\tilde{\pi}=\mathcal G(\bar\beta^{(L)},\bar\beta^{(H)},\phi)$.
The KLD between $\pi$ and $\tilde{\pi}$ is given by
\begin{align*}
\!\!\!\mathcal D(\pi,\bar\beta^{(L)},\bar\beta^{(H)},\phi)\triangleq D(\pi||\tilde{\pi})=
\sum_{\mathbf b\in\{0,1\}^F}\pi(\mathbf b)\ln\left(\frac{\pi(\mathbf b)}{\tilde{\pi}(\mathbf b)}\right).\!\!
\end{align*}
The goal is, given $\pi$, to find parameters $(\bar\beta^{(L)*},\bar\beta^{(H)*},\phi^*)(\pi)$
such that
\begin{align}
\label{P1}
&(\bar\beta^{(L)*},\bar\beta^{(H)*},\phi^*)(\pi){=}\underset{\bar\beta^{(L)},\bar\beta^{(H)},\phi}{\arg\min}\ \mathcal D(\pi,\bar\beta^{(L)},\bar\beta^{(H)},\phi)
\!\!
\\&\nonumber
=\underset{\bar\beta^{(L)},\bar\beta^{(H)},\phi}{\arg\max}
\sum_i
\left[
{\boldsymbol{\beta}}_{i}\ln \left(\bar\beta^{(\phi(i))}\right)
{+}(1{-}{\boldsymbol{\beta}}_{i})\ln \left(1{-}\bar\beta^{(\phi(i))}\right)
\right],
\end{align}
where we have used (\ref{approxbeilef}) and defined
 ${\boldsymbol{\beta}}{=}\mathbb E[\mathbf b_{k}|\pi_k{=}\pi]$.
 Theorem \ref{thm1} determines the solution of~(\ref{P1}).
%{-3mm}
\vspace{-10mm}
\begin{thm}
\label{thm1}
The solution of  (\ref{P1}) is given by
\begin{align}
&\bar\beta^{(i)*}(\pi)=\bar\beta^{(i)}(\nu^*(\pi)),\ i\in\{L,H\},
\ 
&\phi^*(\pi;m(i))=\left\{\begin{array}{ll}L,&i\leq\nu^*(\pi),\\H,&i>\nu^*(\pi),\end{array}\right.
\end{align}
%Given  $\boldsymbol{\beta}=\mathbb E[\mathbf b|\pi]$,
%the parameters of the  approximate belief, solution of (\ref{P1}), are given by
where% $\bar\beta^{(i)}(\nu),i{\in}\{L,H\}$, are defined as
 \begin{align}
 \label{x1}
&\bar\beta^{(L)}(\nu)=\frac{1}{\nu}\sum_{i=1}^{\nu}\boldsymbol{\beta}_{m(i)},
%\\& \label{x2}
\  \bar\beta^{(H)}(\nu)=\frac{1}{F-\nu}\sum_{i=\nu+1}^{F}\boldsymbol{\beta}_{m(i)},
 \end{align}
  $m:\{1,2,\dots, F\}\mapsto\{1,2,\dots, F\}$ is a permutation of the entries of 
 $\boldsymbol{\beta}$ in increasing order, \emph{i.e.}, such that
 $\boldsymbol{\beta}_{m(1)}\leq \boldsymbol{\beta}_{m(2)}\leq\dots\leq \boldsymbol{\beta}_{m(F)}$, 
 and $\nu^*(\pi)$ solves
\begin{align}
\label{nuopt}
&\!\!\!\nu^*(\pi)=\!\!\!\!
\underset{\nu\in\{1,2,\dots, F-1\}}{\arg\min}
\nu H_2\left(\bar\beta^{(L)}(\nu)\right)
+(F-\nu)H_2\left(\bar\beta^{(H)}(\nu)\right),
\end{align}
where $H_2(x){=}{-}x\ln \left(x\right){-}(1{-}x)\ln \left(1{-}x\right)$ is the binary entropy function.
\end{thm}
\begin{proof}
See Appendix C.
\end{proof}

A sufficient statistic to represent the approximate belief $\tilde\pi$
is the  \emph{compressed belief state} (CBS)
$\mathbf s{=}(\bar\beta^{(L)},\bar\beta^{(H)},\nu)$, where $\nu{\in}\{1,2,\dots,F-1\}$ is the number of bands detected as idle,
$P_{FA}{=}1{-}\bar\beta^{(H)}$ and
$P_{MD}{=}\bar\beta^{(L)}$ are
 the false-alarm and missed-detection probabilities for the bands detected as busy and idle, respectively.
In fact, any $\phi(\cdot)$ which maps $\nu$ spectrum bands to the low probability of occupancy $\bar\beta^{(L)}$
and the remaining $F{-}\nu$ spectrum bands  to the high probability of occupancy $\bar\beta^{(H)}$ can be obtained by a proper permutation
of the spectrum bands, which preserves the dynamics of the system, 
due to the symmetry of the spectrum bands across frequency.
Thus, the specific $\phi(\cdot)$ needs not be taken into account, but only the number of bands detected as idle, $\nu$.
We denote the projection operator as
$\mathbf s{=}\mathcal P(\boldsymbol{\beta})$ or $\mathbf s{=}\mathcal P(\pi)$ (used interchangeably).

Therefore, given the prior belief $\pi_k$ and  posterior belief $\hat\pi_k$, we denote the \emph{prior} CBS as 
$\mathbf s_k=\mathcal P(\pi_k)$
and the \emph{posterior} CBS as 
$\hat{\mathbf s}_k=\mathcal P(\hat\pi_k)$, determined as in Theorem~\ref{thm1}.
We then define the policy  $\psi_k=\psi(\mathbf s_k)$
which maps the prior CBS to a value of the sensing traffic $\psi_k$,
and the policy $\Lambda_k=\Lambda(\hat{\mathbf s}_k)$,
which maps the posterior CBS to a value of the total traffic budget $\Lambda_k$.
While the prior and posterior beliefs $\pi_k$ and $\hat\pi_k$ are probability distributions over a space of size $2^F$ (all the possible realizations of the
spectrum occupancy vector $\mathbf b_k$), which scales exponentially with the spectrum size $F$,
 the CBS takes value from a low-dimensional space, which scales linearly with $F$.
Therefore, $\psi(\mathbf s_k)$ and $\Lambda(\hat{\mathbf s}_k)$ can be found with lower complexity than
 $\psi(\pi_k)$ and $\Lambda(\hat\pi_k)$.

Despite the dimensionality reduction achieved by operating based on the  CBS,
computing $\hat\pi_k=\hat\Pi(\tilde\pi_k,\mathbf A_k,\mathbf y_k)$
and $\hat{\boldsymbol{\beta}}_k$
 in the sensing phase via (\ref{posterior})
has exponential complexity. To achieve complexity reduction,  we propose to decouple the estimator from the CC, \emph{i.e.},
the estimator is treated as a black-box 
 with input $(\tilde\pi_k,\mathbf A_k,\mathbf y_k)$, 
which outputs a maximum-a-posteriori (MAP) estimate $\hat{\mathbf b}_k^{(MAP)}$ of $\mathbf b_k$ (Sec. \ref{MAP}) and 
 \emph{posterior false-alarm}
and  \emph{missed-detection probabilities}
for the bands detected as busy and idle, denoted as  $\hat P_{FA,k}$ and $\hat P_{MD,k}$, respectively.
Given $\hat{\mathbf b}_k^{(MAP)}$, $\hat P_{FA,k}$ and $\hat P_{MD,k}$,
 the CC approximates the posterior expected occupancy as
\begin{align}
&\hat{\boldsymbol{\beta}}_{k}=
\hat{\mathbf b}_{k}^{(MAP)}(1-\hat P_{FA,k})+(1-\hat{\mathbf b}_{k}^{(MAP)})\hat P_{MD,k},
\end{align}
from which the CBS $\hat{\mathbf s}_k=(\hat{\bar\beta}_k^{(L)},\hat{\bar\beta}_k^{(H)},\hat \nu_k)$
is determined as
 $\hat{\bar\beta}_k^{(L)}=\hat P_{MD,k}$, $\hat{\bar\beta}_k^{(H)}=1-\hat P_{FA,k}$ and $\hat \nu_k=F-\sum_i\hat{\mathbf b}_{k,i}^{(MAP)}$,
 and the mapping function $\hat\phi_k$ as $\hat\phi_k(i)=L\Leftrightarrow \hat{\mathbf b}_{k,i}^{(MAP)}=0$.
\vspace{-5mm}
\subsection{Spectrum estimation via sparse recovery}
\label{MAP}
Given the prior $\boldsymbol{\beta}_k=\mathbb E[\mathbf b_k|\pi_k]$, $\hat{\mathbf b}_k^{(MAP)}$ solves
\begin{align}
\label{mapopt}
&\hat{\mathbf b}_k^{(MAP)}=\underset{\mathbf b\in\{0,1\}^F}{\arg\max}\ \mathbb P(\mathbf b_k=\mathbf b|\boldsymbol\beta_k,\mathbf A_k,\mathbf y_k)
%\\&
=
\underset{\mathbf b\in\{0,1\}^F}{\arg\min} \left\Vert\mathbf y_k-\mathbf A_k^T\mathbf b\right\Vert_F^2+2\sigma_Z^2\sum_i\mathbf b_{i}\ln\left(\frac{1-\boldsymbol\beta_{k,i}}{\boldsymbol\beta_{k,i}}\right),\nn
\end{align}
where we have assumed the factorized prior distribution
\begin{align}
 \mathbb P(\mathbf b_k=\mathbf b|\boldsymbol\beta_k)=\prod_i\boldsymbol\beta_{k,i}^{\mathbf b_i}(1-\boldsymbol\beta_{k,i})^{1-\mathbf b_i}.
\end{align}
In particular,
letting ${\mathbf b}_k^{(map)}=\chi(\boldsymbol{\beta}_k\geq 0.5)$ be the \emph{maximum-a-priori} estimate,
 we can rewrite
\begin{align}
\label{asdfsdfh}
\hat{\mathbf b}_k^{(MAP)}={\mathbf b}_k^{(map)}\oplus\hat{\mathbf e}_k^{(MAP)},
\end{align}
where $\hat{\mathbf e}_k^{(MAP)}$ is the correction vector informed by the measurement matrix $\mathbf A_k$ and observation vector $\mathbf y_k$.
By plugging (\ref{asdfsdfh}) into (\ref{mapopt}),
this is
given by the solution of the optimization problem
\begin{align}
\label{MAPe}
&\hat{\mathbf e}_k^{(MAP)}=
\underset{\mathbf e\in\{0,1\}^F}{\arg\min}\left\Vert
\hat{\mathbf y}_k
-\hat{\mathbf A}_k^{T}{\mathbf e}
\right\Vert_F^2
+\boldsymbol{\mu}_k^T
{\mathbf e},
\end{align}
where we have defined
\begin{align}
&\hat{\mathbf y}_k\triangleq\mathbf y_k-\mathbf A_k^T\mathbf b_{k}^{(map)},
\ 
\hat{\mathbf A}_k\triangleq \left(\mathbf I_F-2\mathrm{diag}(\mathbf b_{k}^{(map)})\right)\mathbf A_k,
\end{align}
as the  residual error from the prior estimate and the  corrected measurement matrix, respectively,
and $\boldsymbol{\mu}_k$ is a Lagrangian multiplier column vector with components
\begin{align}
\boldsymbol\mu_{k,i}\triangleq 2\sigma_Z^2\left(1-2\mathbf b_{k,i}^{(map)}\right)\ln\left(\frac{1-\boldsymbol\beta_{k,i}}{\boldsymbol\beta_{k,i}}\right).
\end{align}
Note that the Lagrangian vector $\boldsymbol{\mu}_k$ weights the components of the error vector $\mathbf e$ based on their prior log-likelihood.
As a result, each component $\mathbf e_i$ may be weighted in a different way, according to its prior.
Moreover, from the definition of prior estimate $\mathbf b_{k,i}^{(map)}$, we have that $\boldsymbol\mu_{k,i}\geq 0$, with equality if and only if $\boldsymbol\beta_{k,i}=0.5$.

The optimization problem (\ref{MAPe}) has combinatorial complexity, since the cost function needs to be evaluated for each $\mathbf e\in\{0,1\}^F$.
In order to overcome the combinatorial complexity, we propose the following convex $l_1$ relaxation:
\begin{align}
\label{MAPe2}
&\tilde{\mathbf e}_k=
\underset{\mathbf e\in[0,1]^F}{\arg\min}\left\Vert
\hat{\mathbf y}_k
-\hat{\mathbf A}_k^{T}{\mathbf e}
\right\Vert_F^2
+\boldsymbol{\mu}_k^T
{\mathbf e},
\end{align}
\emph{i.e.}, the optimization is over the convex set $[0,1]^F$, rather than the discrete one $\{0,1\}^F$,
and can thus be solved using convex optimization techniques~\cite{Boyd}.
In particular, it is a  quadratic programming problem minimizing a least-squares term, plus an $\ell_1$ regularization term,
which induces sparsity in the correction vector $\tilde{\mathbf e}_k$. The larger $\boldsymbol\mu_{k,i}$ (\emph{i.e.}, the closer $\boldsymbol\beta_{k,i}$ to $0$ or $1$), the sparser the solution.
Note that $\tilde{\mathbf e}_k$ is not feasible with respect to the original optimization problem (\ref{MAPe}).
A feasible point is thus obtained by projecting $\tilde{\mathbf e}_k$ into the discrete set $\{0,1\}^F$
 using, \emph{e.g.}, the minimum distance criterion $\chi(\tilde{\mathbf e}_k\geq 0.5)$.
This solution is not globally optimal with respect to (\ref{MAPe}), and can be improved using the following 
hill climbing algorithm~\cite{Russell}.
\vspace{-5mm}
\begin{algo}[Hill climbing algorithm]% $ $\\
\noindent {\bf 1)} \emph{Initialization}: $\hat{\mathbf e}_k^{[0]}=\chi(\tilde{\mathbf e}_k\geq 0.5)$, counter $l=0$;

\noindent {\bf 2)} \emph{Improvement step}: at step $l$, compute the vector $\boldsymbol{\Delta}^{[l]}$, with the $i$th component given by
\begin{align}
%f-fnew
\nn
&\boldsymbol{\Delta}_i^{[l]}=
\left(2\hat{\mathbf e}_{k,i}^{[l]}-1\right)
\left\{
%\vphantom{+2\sum_{j\neq i}[\hat{\mathbf A}_k\hat{\mathbf A}_k^T]_{i,j}\hat{\mathbf e}_{k,j}^{[l]}}
- 2[\hat{\mathbf A}_k\hat{\mathbf y}_k]_i
%\right.\\&\left.\qquad
+2\sum_{j\neq i}[\hat{\mathbf A}_k\hat{\mathbf A}_k^T]_{i,j}\hat{\mathbf e}_{k,j}^{[l]}
+\left[\hat{\mathbf A}_k\hat{\mathbf A}_k^T\right]_{i,i}
+\boldsymbol{\mu}_{k,i}
\right\}.
\end{align}
Let $i^*=\arg\max_i \boldsymbol{\Delta}_i^{[l]}$;
if $\boldsymbol{\Delta}_{i^*}^{[l]}\geq 0$,
determine a new MAP estimate $\hat{\mathbf e}_k^{[l+1]}$ as
$\hat{\mathbf e}_{k,i}^{[l+1]}=\hat{\mathbf e}_{k,i}^{[l]},\ \forall i\neq i^*$,
$\hat{\mathbf e}_{k,i^*}^{[l+1]}=1-\hat{\mathbf e}_{k,i^*}^{[l]}$,
update the counter $l:=l+1$ and repeat from the \emph{improvement step};
otherwise, return the MAP estimate $\hat{\mathbf e}_k^{(MAP)}:=\hat{\mathbf e}_{k}^{[l]}$.
\end{algo}
The term $\boldsymbol{\Delta}_i^{[l]}$ represents the increase or decrease in the MAP cost function (\ref{MAPe2}), by 
switching the $i$th component of the current MAP estimate, $\hat{\mathbf e}_i^{[l]}$, from 1 to 0, or vice versa,
and keeping all the other components unchanged.
In particular, $\boldsymbol{\Delta}_i^{[l]}$ is the difference in the  cost function (\ref{MAPe2}) between the 
old cost and the new one, so that $\boldsymbol{\Delta}_{i^*}^{[l]}>0$ if an improved estimate is obtained. 
By the definition of $i^*$, if $\boldsymbol{\Delta}_{i^*}^{[l]}\geq 0$, by switching the $i^*$th band of the current MAP estimate, 
the MAP cost function (\ref{MAPe2}) is decreased by the amount $\boldsymbol{\Delta}_{i^*}^{[l]}$,
yielding an improved estimate. Otherwise ($\boldsymbol{\Delta}_{i^*}^{[l]}<0$),
a local optimum has been determined by the algorithm, \emph{i.e.},
any change of one and only one component of the current MAP estimate is sub-optimal.
\vspace{-5mm}
\subsection{CBS transition probabilities}
\label{txprob}
In order to run the DP algorithm based on the CBS, we need to determine the corresponding transition probabilities.
Note that $\hat{\mathbf b}_k^{(MAP)}$ can be written as a
function of the prior expected occupancy $\boldsymbol\beta_{k}$, which maps to the corresponding CBS $\mathbf s_k=\mathcal P(\boldsymbol\beta_{k})$,
and of $(\mathbf A_k,\mathbf y_k)$. We denote this function as 
$\hat{ \mathbf b}_{k}^{(MAP)}=\mathrm{MP}(\boldsymbol\beta_{k},\mathbf A_k,\mathbf y_k)$.
Similarly, the false-alarm and missed-detection probabilities can be written as
functions of $(\boldsymbol\beta_{k},\mathbf A_k,\mathbf y_k)$, and thus
are of difficult evaluation, due to their dependence on the measurements $(\mathbf A_k,\mathbf y_k)$.

Herein, we propose to
marginalize
the false-alarm and missed-detection probabilities
 with respect  the measurements $(\mathbf A_k,\mathbf y_k)$,
for a given value of
the number of bands detected as idle,
$\hat\nu_k=F-\sum_i\mathrm{MP}_i(\boldsymbol\beta_{k},\mathbf A_k,\mathbf y_k)$,
 the number of measurements received, $M_k$, and the CBS $\mathbf s_k$.
 The rationale is that the detection performance of the MAP estimator is mainly driven by the number of measurements collected at the CC,
 rather than the specific observations $(\mathbf A_k,\mathbf y_k)$.
 Equivalently,
\begin{align}
\label{PFA}
&\hat P_{FA}(\mathbf s,m,\hat\nu){\triangleq}
\mathbb E\!\left[\hat P_{FA,k}\!\left|\!\!\!\begin{array}{l}
(\mathbf s_{k},M_k,\hat\nu_k)\\{=}(\mathbf s,m,\hat\nu)
\end{array}\right.\!\!\!\!
\right]
%\\&
\!\!{=}
\mathbb E\!\!\left[
\frac{\sum_{i}(1{-}\mathbf b_{k,i})\mathrm{MP}_i(\boldsymbol\beta_k,\mathbf A_k^{(m)},\mathbf y_k^{(m)})}{F-\hat\nu}
\left|
\!\!\!\begin{array}{l}
(\mathbf s_{k},M_k,\hat\nu_k)\\{=}(\mathbf s,m,\hat\nu)
\end{array}\!\!\!
\right.\right]\!\!,\!
%\nn
\\&
\label{PMD}
\!\!\hat P_{MD}(\mathbf s,m,\hat\nu)
{\triangleq}
\mathbb E\!\left[\!\hat P_{MD,k}\!\left|
\!\!\!\begin{array}{l}
(\mathbf s_{k},M_k,\hat\nu_k)\\{=}(\mathbf s,m,\hat\nu)
\end{array}\!\!\!\!
\right.\right]
\!\!{=}
\mathbb E\!\left[\!\left.
\frac{\sum_{i}\mathbf b_{k,i}(1{-}\mathrm{MP}_i(\boldsymbol\beta_k,\mathbf A_k^{(m)},\mathbf y_k^{(m)}))}{\hat\nu}
\right|
\!\!\!\begin{array}{l}
(\mathbf s_{k},M_k,\hat\nu_k)\\{=}(\mathbf s,m,\hat\nu)
\end{array}\!\!\!
\right]\!\!,\!\!%\nn
\end{align}
where $\boldsymbol\beta_k$ maps to the CBS $\mathbf s=\mathcal P(\boldsymbol\beta_k)$, up to a random permutation of its entries.
The expectation is taken with respect to the realization of the measurement matrix $\mathbf A^{(m)}\in\mathbb R^{F\times m}$,
the measurement vector $\mathbf y^{(m)}$ (of size $m$),
 and the random permutation of the entries of $\boldsymbol\beta_k$.

Let $P_{\nu}(\hat\nu|\mathbf s,m){\triangleq}\mathbb P(\hat\nu_k{=}\hat\nu|\mathbf s_{k}{=}\mathbf s,M_k{=}m)$ be the pmf of the number of bands detected as idle,
given the prior CBS $\mathbf s$ and the number of measurements received $m$,
 after marginalization with respect to $(\mathbf A_k,\mathbf y_k)$.
 This is given by
 \begin{align}
 \label{pv}
& P_{\nu}(\hat\nu|\mathbf s,m)
% \\&\nn
 =
\mathbb E\left[\left.\chi\left(F{-}\sum_{i}\mathrm{MP}_i(\boldsymbol\beta_k,\mathbf A_k^{(m)},\mathbf y_k^{(m)}){=}\hat\nu\right)\right|
\begin{array}{l}
(\mathbf s_{k},M_k)\\{=}(\mathbf s,m)
\end{array}
\right].
 \end{align}

We define a neighborhood around the (prior or posterior) CBS $(\bar\beta^{(L)},\bar\beta^{(H)},\nu)$ as
 \begin{align}
\mathcal S_\delta(\bar\beta^{(L)},\bar\beta^{(H)},\nu)
\equiv&\left\{\mathbf s=(x,y,\nu):x\in[\bar\beta^{(L)}-\delta,\bar\beta^{(L)}+\delta],
%\right.\nn\\&\left.
y\in[\bar\beta^{(H)}-\delta,\bar\beta^{(H)}+\delta]\right\}.
\end{align}
The transition probability from the prior CBS $\mathbf s_k{=}\mathbf s$ to 
the posterior CBS $\hat{\mathbf s}_k\in\mathcal S_\delta(\hat{\mathbf s})$
is given by
\begin{align}
\label{priortopost}
&\mathbb P(\hat{\mathbf s}_k\in\mathcal S_\delta(\hat{\mathbf s})|\mathbf s_k=\mathbf s,\psi_k=\psi)
=\sum_{m=0}^B
P_M(m|\psi)P_{\nu}(\hat\nu|\mathbf s,m)
\\&\nn
\times\chi\left(
\hat P_{FA}(\mathbf s,m,\hat\nu)\in[1-\hat{\bar\beta}^{(H)}-\delta,1-\hat{\bar\beta}^{(H)}+\delta]
\right)
%\nn\\&\times
\chi\left(
\hat P_{MD}(\mathbf s,m,\hat\nu)\in[\hat{\bar\beta}^{(L)}-\delta,\hat{\bar\beta}^{(L)}+\delta]
\right).
\end{align}

Given the posterior expected occupancy $\hat{\boldsymbol\beta}_k$, the SU traffic $\mathbf r_k$,
and the feedback $\mathbf p_k$,
 the prior expected occupancy in the next slot, $\boldsymbol\beta_{k+1}$,  is given by
  \begin{align}
&\boldsymbol\beta_{k+1,i}=
 \mathbb P(\mathbf b_{k+1,i}=1|\mathbf r_{k,i}=\mathbf r_{i},\mathbf p_{k,i}=\mathbf p_{i},\hat{\boldsymbol\beta}_{k,i})
 \nn\\&
% =
%\frac{\left[\begin{array}{l}\sum_{b\in\{0,1\}}\mathbb P(\mathbf b_{k+1,i}{=}1|\mathbf b_{k,i}{=}b,\mathbf r_{k,i}{=}\mathbf r_{i},\mathbf p_{k,i}{=}\mathbf p_{i})\\\times
%\mathbb P(\mathbf p_{k,i}{=}\mathbf p_{i}|\mathbf b_{k,i}{=}b,\mathbf r_{k,i}{=}\mathbf r_{i})
%\mathbb P(\mathbf b_{k,i}{=}b|\hat{\boldsymbol\beta}_{k,i}) \end{array}\right]}{
%\sum_{b\in\{0,1\}}  \mathbb P( \mathbf p_{k,i}{=}\mathbf p_{i}|\mathbf b_{k,i}{=}b,\mathbf r_{k,i}{=}\mathbf r_{i})
%\mathbb P(\mathbf b_{k,i}{=}b|\hat{\boldsymbol\beta}_{k,i})} \nn\\&
 =
\frac{
%\left[\begin{array}{l}
P_{B|P}(1|0,\mathbf r_{i},\mathbf p_{i})P_{P}(\mathbf p_{i}|0,\mathbf r_{i})(1-\hat{\boldsymbol\beta}_{k,i})%\\
+P_{B|P}(1|1,\mathbf r_{i},\mathbf p_{i})P_{P}(\mathbf p_{i}|1,\mathbf r_{i})\hat{\boldsymbol\beta}_{k,i}
%\end{array}\right]
 }
 {
 P_{P}(\mathbf p_{i}|0,\mathbf r_{i})(1-\hat{\boldsymbol\beta}_{k,i})
 +P_{P}(\mathbf p_{i}|1,\mathbf r_{i})\hat{\boldsymbol\beta}_{k,i}
 },
 \end{align}
 where $P_{B|P}(\cdot)$ and $P_{P}(\cdot)$ are defined in (\ref{pp1})-(\ref{pp3}) and (\ref{ppp1})-(\ref{pppend}),
 so that we can write $\boldsymbol\beta_{k+1}=\beta(\hat{\boldsymbol\beta}_{k},\mathbf r_{k},\mathbf p_{k})$
 for a proper function $\beta(\cdot)$.
 This, in turn, maps to the prior CBS $\mathbf s_{k+1}=\mathcal P(\boldsymbol\beta_{k+1})$.
 Therefore, for a given total traffic budget $\Lambda_k=\Lambda$, the transition probability from the posterior CBS $\hat{\mathbf s}_k=\hat{\mathbf s}=(\hat{\bar\beta}^{(L)},\hat{\bar\beta}^{(H)},\nu)$ to the prior CBS
 $\mathbf s_{k+1}\in\mathcal S_\delta(\mathbf s)$ in the next slot,
 is given by
  \begin{align}
  \label{posttoprior}
&\mathbb P(\mathbf s_{k+1}\in\mathcal S_\delta(\mathbf s)|\hat{\mathbf s}_k=\hat{\mathbf s},\Lambda_k=\Lambda)
\\&\nn
=
\sum_{\mathbf b,\mathbf p\in\{0,1\}^F}
\prod_i  P_{P}(\mathbf p_{i}|\mathbf b_{i},\Lambda\mathbf z_i(\hat{\boldsymbol\beta},\Lambda))
\hat{\boldsymbol\beta}_{i}^{\mathbf b_i}(1-\hat{\boldsymbol\beta}_{i})^{1-\mathbf b_i}
%\\&\quad\times
\chi\left(\mathcal P\left(\beta\left(\hat{\boldsymbol\beta},\Lambda\mathbf z(\hat{\boldsymbol\beta},\Lambda),\mathbf p\right)\right)\in\mathcal S_\delta(\mathbf s)\right),
\nn
\end{align}
where we have marginalized with respect to $\mathbf b_k$ and $\mathbf p_k$,
and $\hat{\boldsymbol\beta}$ is given by $\hat{\boldsymbol\beta}_i=\hat{\bar\beta}^{(L)},i\leq \hat\nu$,
$\hat{\boldsymbol\beta}_i=\hat{\bar\beta}^{(H)},i>\hat\nu$, so that $\mathcal P(\hat{\boldsymbol\beta})=\hat{\mathbf s}$.

These probabilities, along with (\ref{PFA}), (\ref{PMD}) and  (\ref{pv}), do not admit a closed form analytical expression, but can be computed numerically via Monte-Carlo simulation.
\vspace{-5mm}
\section{Numerical Results}\label{sec:numres}
In this section, we present numerical results for a system with
 parameters: number of frequency bands $F{=}20$; 
SU and  PU failure probabilities $\rho_S{=}\rho_P{=}0.1$; probability of new PU arrival $\zeta{=}0.095$;
probability of a new data packet for an active PU $\theta{=}0.95$;
number of SUs $N_S{=}100$; number of  control channels for the SUs $B{=}5$;
 SU sensing and data transmission costs $c_S{=}c_{TX}{=}1$;
variance of the entries of the measurement matrix $\sigma_A^2{=}1$;
variance of the measurement noise $\sigma_Z^2{=}1/20$; erasure probability $\epsilon{=}0.9$.
The performance is evaluated over $2\times 10^4$ slots.

Fig.~\ref{figSUvsPU}.a plots the trade-off between the PU and SU throughputs,
for different values of the total cost $\bar C$ (accounting for both the cost of sensing and of data transmission).
Fig.~\ref{figSUvsPU}.b plots the fraction of the total cost  $\bar C$  that is spent for spectrum sensing ($\bar C_{sensing}$),
as a function of the PU throughput and total cost $\bar C$. Note that the remaining fraction of the total cost is spent for actual data transmission.
 We notice that the throughput trade-off improves for higher values of the cost $\bar C$.
This is expected since, when more resources are available (higher $\bar C$), there are more opportunities to perform data transmission for the SUs.
At the same time, for higher values of  the cost $\bar C$, a larger fraction of this cost is spent for spectrum sensing.
The reason is that, in order to accommodate more traffic for the SUs, with minimal interference to the PUs, the SUs need to acquire more accurate spectrum estimates,
hence the sensing cost increases accordingly.
Additionally, when the requirement on the throughput degradation to the PUs is very strict ($\bar T_P$ approaching the maximal value),
most of the resources are spent for spectrum sensing.
This is because, in order to meet the demanding requirement on the throughput degradation to the PU, the SU traffic
should be allocated only on those spectrum bands which are idle almost surely. Such
low level of uncertainty, in turn, demands significant sensing resources.

\begin{figure}[t]
    \centering
    \subfigure[Trade-off between PU and SU throughput for different values of the cost constraint $\bar C$ (accounting both sensing and data transmission costs).]
    {
\includegraphics[width=.45\linewidth,trim = 1mm 1mm 1mm 1mm,clip=true]{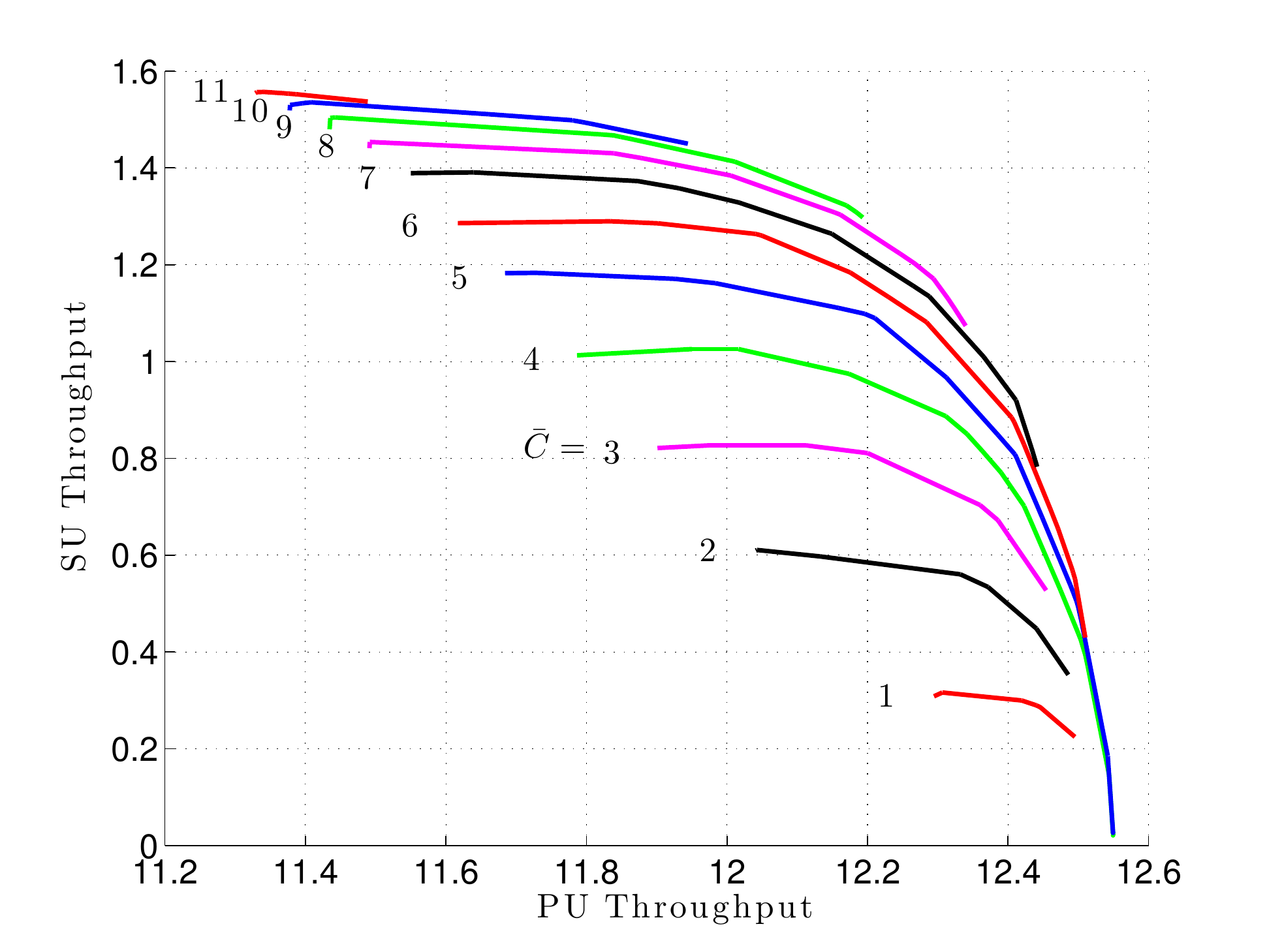}
}
    \subfigure[Fraction of the total cost $\bar C$ spent for spectrum sensing ($\bar C_{sensing}$),
for different values of the PU throughput and total cost $\bar C$.]
    {
\includegraphics[width=.45\linewidth,trim = 1mm 1mm 1mm 1mm,clip=true]{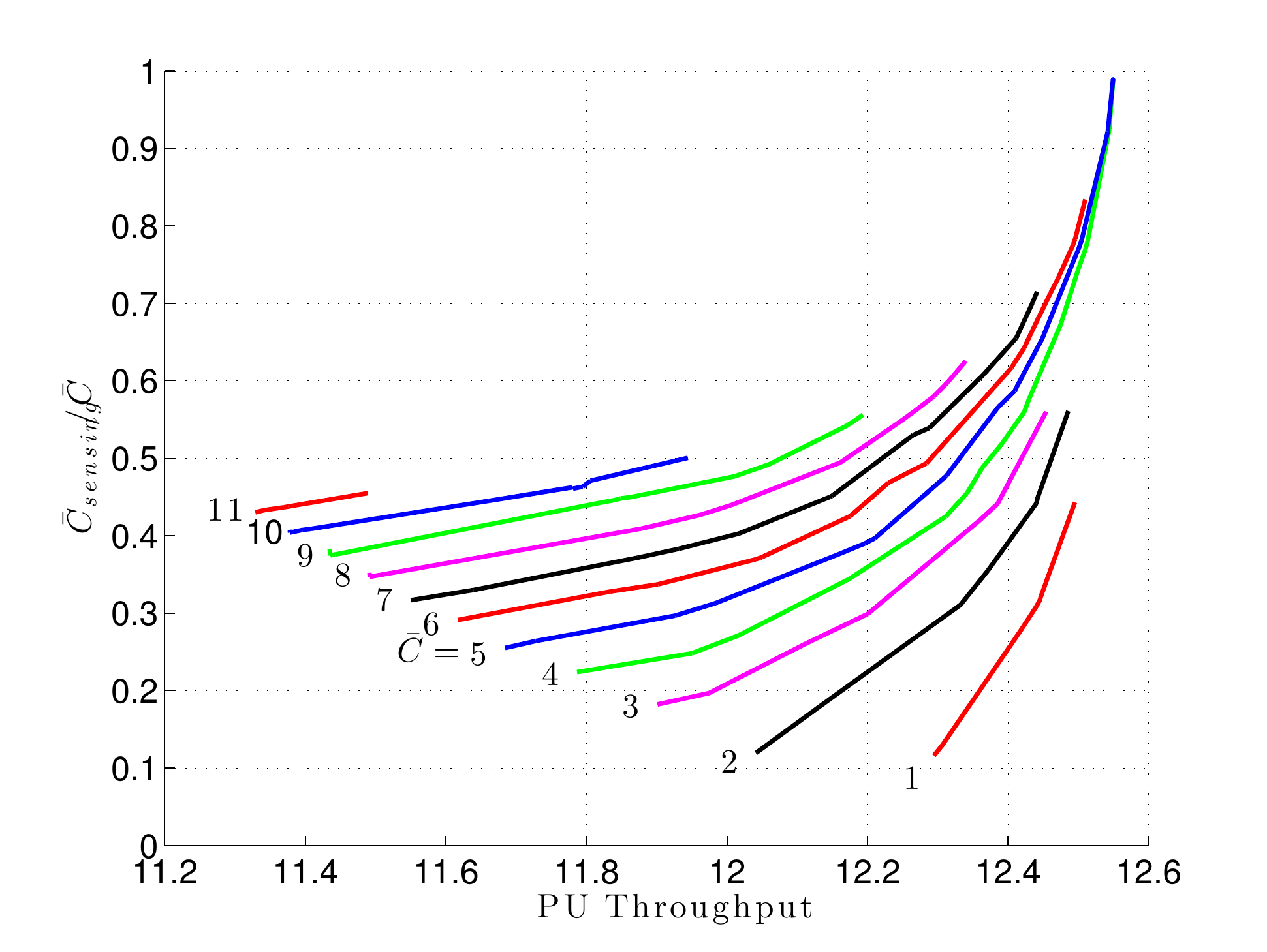}
}
\caption{}
\label{figSUvsPU}
\vspace{-9mm}
\end{figure}

Fig. \ref{samples}.a plots the total traffic allocated, $\Lambda_k$, as a function of the expected number of 
   occupied spectrum bands, $\sum_i\hat{\boldsymbol\beta}_{k,i}$, for the case $\lambda=0.025$ and $\xi=0.7$.
   Each sample corresponds to a given value of the posterior belief $\hat{\boldsymbol\beta}_{k}$, lying on the low-dimensional manifold generated by the 
   CBS.
   We notice that, as the expected number of occupied spectrum bands increases, the total traffic allocated tends to decrease. 
   In fact, when more spectrum bands are expected to be occupied, there are fewer opportunities to occupy the remaining idle bands by the SUs.
   Fig. \ref{samples}.b plots the SU sensing traffic per channel, $\psi_k$, 
    and the expected number of measurements received at CC, $\mathbb E(M_k)$, as a function of the entropy of the prior belief state, $\sum_iH_2(\boldsymbol\beta_{k,i})$,
    where $\boldsymbol\beta_{k}$ lies on the low-dimensional manifold generated by the CBS. We notice that, as the entropy increases, \emph{i.e.}, the amount of uncertainty on the current spectrum occupancy increases, the SU sensing traffic also increases, and thus, more measurements are collected at the CC in order to reduce this uncertainty. 
    The sensing resources are focused in those regions of the belief state where the spectrum occupancy state is more uncertain,
    yielding energy efficient resource utilization.
\begin{figure}[t]
    \centering
    \subfigure[Total traffic allocated $\Lambda_k$ as a function of the expected number of 
   occupied spectrum bands, $\sum_i\hat{\boldsymbol\beta}_{k,i}$.]
    {
\includegraphics[width=.45\linewidth,trim = 1mm 1mm 1mm 8mm,clip=true]{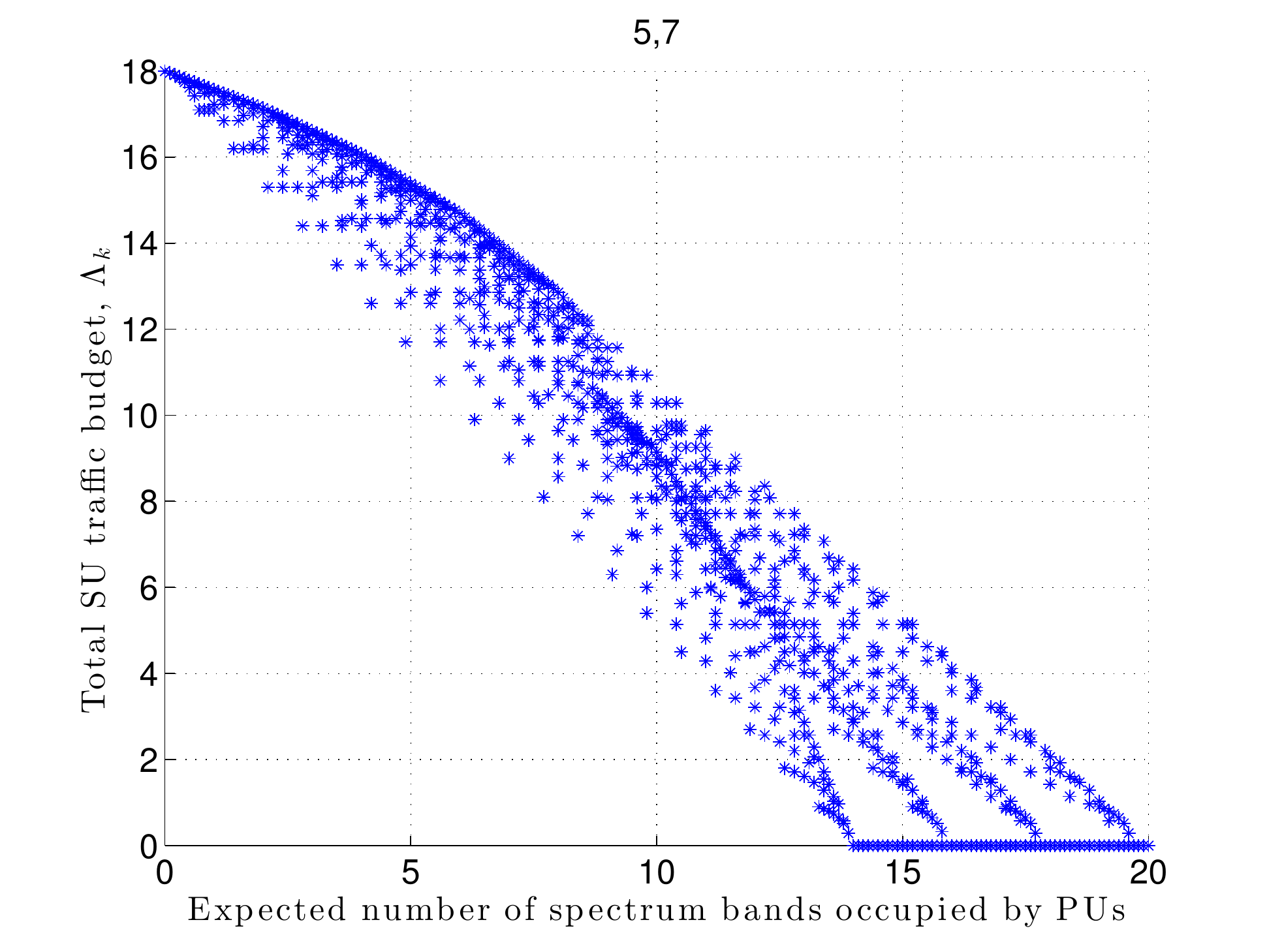}
}
    \subfigure[SU sensing traffic per channel $\psi_k$ and expected number of measurements received at CC $\mathbb E(M_k)$ as a function of the entropy of the prior belief state, $\sum_iH_2(\boldsymbol\beta_{k,i})$.  ]
    {
\includegraphics[width=.45\linewidth,trim = 1mm 1mm 1mm 8mm,clip=true]{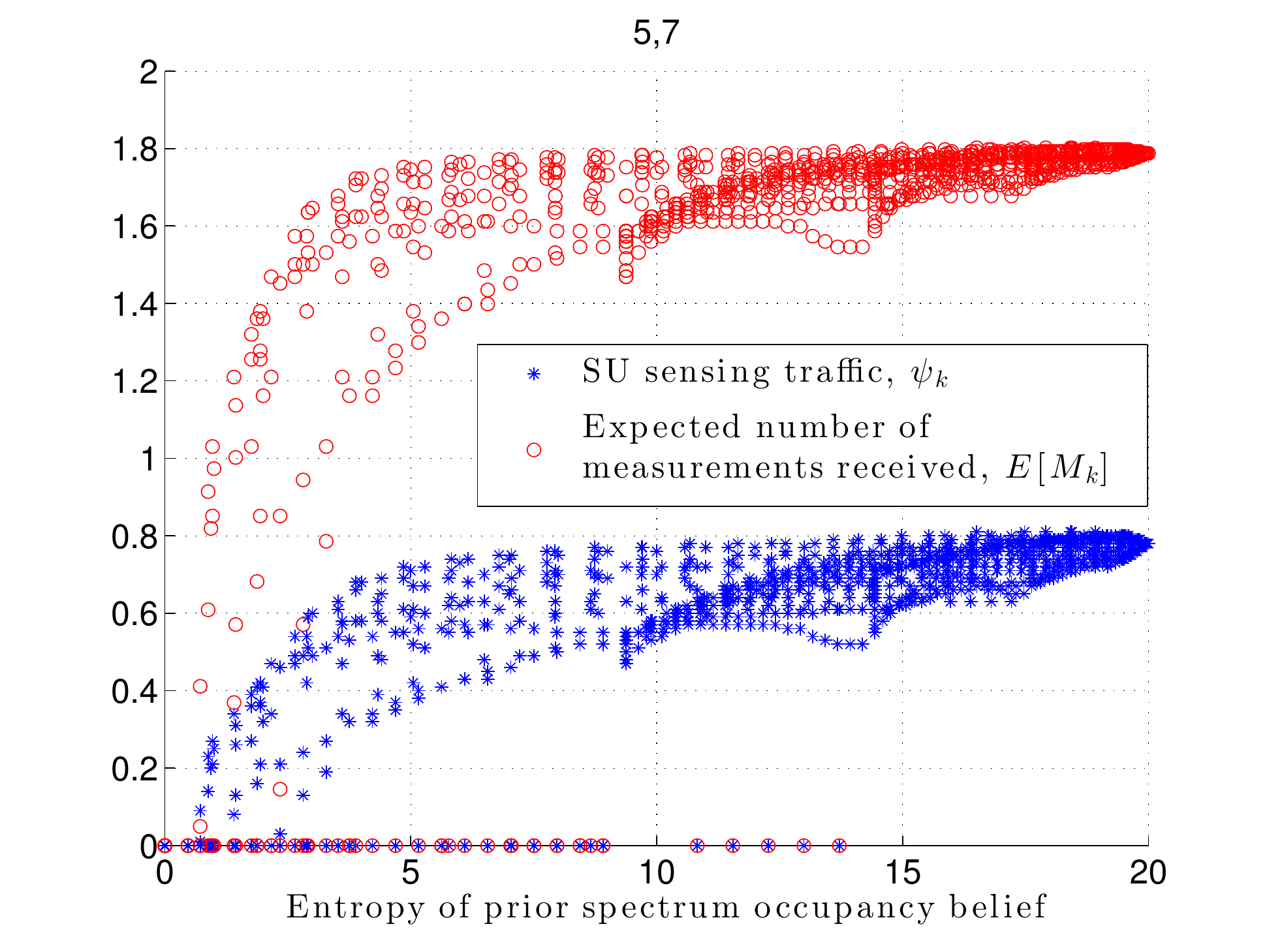}
}
\caption{}
\label{samples}
\vspace{-9mm}
\end{figure}
\vspace{-5mm}
\section{Conclusions}\label{conclu}
In this paper, we have presented a cross-layer framework for joint distributed spectrum sensing, estimation and scheduling in a wireless network
composed of SUs that opportunistically access portions of the spectrum left unused by a licensed network of PUs.
In contrast to much prior work, we jointly address sensing and control, wherein the sensing affects the quality of the measurements.
 Inference of the underlying spectrum occupancy state is obtained by collecting compressed measurements at the CC from nearby SUs, and via local ACK/NACK feedback information from the PUs.
   In order to reduce the huge optimization and operational complexity due to the POMDP formulation, we have proposed
 a technique to project the belief state onto a low-dimensional manifold via the minimization of the Kullback-Leibler divergence.
 We have proved the optimality of  a two-stage decomposition,
 which enables the \emph{decoupling} of the optimization of sensing and scheduling.
Additionally, we have proposed a partially myopic optimization scheme, 
  which can be solved efficiently using convex optimization tools.
    Simulation results demonstrate how the proposed framework optimally balances the cost of acquisition of state information via distributed spectrum sensing and
  the cost of data transmission incurred by the SUs, while achieving the best trade-off between PU and SU throughput under the resource constraints available.
  \vspace{-6mm}
%{-3mm}
\section*{Appendix A: Proof of Theorem \ref{suffstat}}
The instantaneous expected sensing cost (see (\ref{cSensing})) is given by $\psi_kBc_S$,
and is thus independent of $\mathcal H_k$, given $\psi_k$.
Moreover, the distribution of $\mathbf b_k$ given $\mathcal H_k$ is given by
$\mathbb P(\mathbf b_k=\mathbf b|\mathcal H_k)=\pi_k(\mathbf b)$.
Therefore, the prior belief $\pi_k$ is a sufficient statistic to select $\psi_k$ in slot $k$ \cite{POMDP}.

After selecting $\psi_k$ and collecting $(\mathbf A_k,\mathbf y_k)$, the CC 
computes the posterior belief $\hat\pi_k=\hat\Pi(\pi_k,\mathbf A_k,\mathbf y_k)$ as in (\ref{posterior}).
In the scheduling phase, given the SU traffic $\mathbf r_k$ and the history $\hat{\mathcal H}_k$,
the PU feedback $\mathbf p_k$ has probability distribution
\begin{align}
&\mathbb P\left(\left.\mathbf p_k{=}\mathbf p\right|\mathbf r_k,\hat{\mathcal H}_k\right)
%\\&%\times
=\!\!\!\!\!\!
\sum_{\mathbf b\in\{0,1\}^F}\!\!\!\!\!\!
\mathbb P\left(\left.\mathbf p_k{=}\mathbf p\right|\mathbf b_k{=}\mathbf b,\mathbf r_k,\hat{\mathcal H}_k\right)
\mathbb P\left(\left.\mathbf b_k{=}\mathbf b\right|\mathbf r_k,\hat{\mathcal H}_k\right)
\nonumber\\&%\times
=\!\!\!\!\!\!
\sum_{\mathbf b\in\{0,1\}^F}
\prod_i
 P_{P}(\mathbf p_{i}|\mathbf b_{i},\mathbf r_{k,i})
\hat\pi_k(\mathbf b)
=\mathbb P\left(\left.\mathbf p_k=\mathbf p\right|\mathbf r_k,\hat\pi_k\right),
\nonumber
\end{align}
where we have used (\ref{Distrp}) and the definition of $\hat\pi_k$,
thus the distribution is independent of $\hat{\mathcal H}_k$ given $(\mathbf r_k,\hat\pi_k)$.
Therefore, the next prior belief $\pi_{k+1}=\Pi(\hat\pi_k,\mathbf r_k,\mathbf p_k)$
is statistically independent of $\hat{\mathcal H}_k$, given $(\mathbf r_k,\hat\pi_k)$.
The instantaneous expected  data transmission cost (\ref{csched}), and SU/PU throughputs (\ref{rewX}) given $(\mathbf r_k,\hat{\mathcal H}_k)$,
are given by
\begin{align}
&
\mathbb E\left[\left.
c_{TX}\sum_{i=1}^F\mathbf r_{k,i}
\right|\mathbf r_k,\hat{\mathcal H}_k\right]=c_{TX}\sum_{i=1}^F\mathbf r_{k,i},
\\&
\mathbb E\left[\left.T_{X}(\mathbf b_{k},\mathbf r_k)\right|\mathbf r_k,\hat{\mathcal H}_k\right]
%\\&
=\!\!\!\!\!\!\!
\sum_{\mathbf b\in\{0,1\}^F}\!\!\!\!\!
\mathbb P\left(\left.\mathbf b_{k}=\mathbf b\right|\mathbf r_k,\hat{\mathcal H}_k\right)
T_{X}(\mathbf b,\mathbf r_k)
\nn
%\\&
=\!\!\!\!\!\!\!
\sum_{\mathbf b\in\{0,1\}^F}\!\!\!\!\!
\hat\pi_k(\mathbf b)T_{X}(\mathbf b,\mathbf r_k).
\nn
\end{align}
All these metrics of interest are functions of $\mathbf r_k$ and $\hat\pi_k$ only.
Therefore, the posterior belief state $\hat\pi_k$ is a sufficient statistic to schedule the traffic $\mathbf r_k$ in slot $k$ \cite{POMDP}.
\vspace{-5mm}
\section*{Appendix B: Proof of Theorem \ref{thm2}}
%We consider a generic $\hat{\boldsymbol\beta}$, and, for convenience, neglect any functional dependence on $\hat{\boldsymbol\beta}$.\\
\emph{Proof of Property 2)}
From (\ref{asdfasddfdf}), if $\hat{\boldsymbol\beta}_i\geq\frac{\xi(1-\rho_S)}{(1-\xi)(1-\rho_P)+\xi(1-\rho_S)}$.
\\
\emph{Proof of Property 1)}
Consider the optimization problem 
  \begin{align}
  \label{profdsf}
&\max_{\mathbf 0\leq \mathbf r\leq \mathbf r_{\max}}
\sum_{i=1}^F\left[(1{-}\hat{\boldsymbol{\beta}}_{i})\mathbf  r_{i}+\frac{1-\xi}{\xi}\frac{1{-}\rho_P}{1{-}\rho_S}\hat{\boldsymbol{\beta}}_{i} \right]e^{- \mathbf r_{i}}
%\nonumber\\&
\text{ s.t.}\ \sum_i\mathbf r_{i}=\Lambda,
 \end{align}
obtained by replacing $\mathbf r{=}\Lambda\mathbf z$ in (\ref{optz2}).
Let $i$ such that $\hat{\boldsymbol\beta}_i{<}\frac{\xi(1-\rho_S)}{(1-\xi)(1-\rho_P)+\xi(1-\rho_S)}$, hence
$\mathbf r_{\max,i}{>}0$.
The second derivative of the objective function with respect to $\mathbf r_i$ is given by
 \begin{align}
 g(\mathbf  r_{i})\triangleq
\left[(1{-}\hat{\boldsymbol{\beta}}_{i})(\mathbf  r_{i}-2)+\frac{1-\xi}{\xi}\frac{1{-}\rho_P}{1{-}\rho_S}\hat{\boldsymbol{\beta}}_{i} \right]e^{- \mathbf r_{i}}.
 \end{align}
 Since $\mathbf  r_{i}\leq \mathbf r_{\max,i}$, we then obtain
   \begin{align}
 g(\mathbf  r_{i})\leq g(\mathbf r_{\max,i}(\hat{\boldsymbol\beta}))\leq -(1{-}\hat{\boldsymbol{\beta}}_{i})e^{- \mathbf r_{\max,i}(\hat{\boldsymbol\beta})}<0.
 \end{align}
Therefore, the objective function in (\ref{profdsf}) is a concave function of $\mathbf r$.
Since the constraint set $\{\mathbf 0\leq \mathbf r\leq \mathbf r_{\max},\sum_i\mathbf r_{i}=\Lambda\}$
is convex, the resulting optimization problem (\ref{profdsf}) is convex.
\\
\emph{Proof of Property 3)}
We denote the maximizer of (\ref{profdsf}) as $\mathbf r^*(\Lambda)$, which obeys $\mathbf 0\leq \mathbf r^*(\Lambda)\leq \mathbf r_{\max}$.
Solving (\ref{profdsf}) with the Lagrange multiplier method, we obtain
$\max_{\mathbf 0\leq \mathbf r\leq \mathbf r_{\max}} f(\mathbf r,\mu)$,
where
\begin{align}
\label{lagrangian}
f(\mathbf r,\mu)\triangleq&
\sum_{i=1}^F\left[(1{-}\hat{\boldsymbol{\beta}}_{i})\mathbf  r_{i}+\frac{1-\xi}{\xi}\frac{1{-}\rho_P}{1{-}\rho_S}\hat{\boldsymbol{\beta}}_{i} \right]e^{- \mathbf r_{i}}
%\nonumber\\&
+\mu\left(\sum_i\mathbf r_{i}-\Lambda\right),
\end{align}
whose maximizer is denoted as $\tilde{\mathbf r}(\mu)$. The optimal Lagrange multiplier $\mu^*(\Lambda)$ must be such that
\begin{align}
\sum_i\tilde{\mathbf r}_i(\mu^*(\Lambda))=\Lambda,
\end{align}
yielding $\mathbf r^*(\Lambda)=\tilde{\mathbf r}(\mu^*(\Lambda))$.
We now solve the Lagrangian problem (\ref{lagrangian}) for a specific $\mu$.
Since the objective function is a concave function of $\mathbf r:\mathbf 0\leq \mathbf r\leq \mathbf r_{\max}$,
we have the following cases:

\noindent a) If $\left.\frac{\mathrm df(\mathbf r,\mu)}{\mathrm d \mathbf r_i}\right|_{\mathbf r_i=0}\leq 0$, then $\tilde{\mathbf r}_i(\mu)=0$. Equivalently,
\begin{align}
\label{p1}
-(1-\hat{\boldsymbol{\beta}}_{i})+\frac{1-\xi}{\xi}\frac{1{-}\rho_P}{1{-}\rho_S}\hat{\boldsymbol{\beta}}_{i}\geq \mu;
\end{align}

\noindent b) If $\left.\frac{\mathrm df(\mathbf r,\mu)}{\mathrm d \mathbf r_i}\right|_{\mathbf r_i=\mathbf r_{\max,i}}\geq 0$,
then $\tilde{\mathbf r}_i(\mu)=\mathbf r_{\max,i}$.
Equivalently,
\begin{align}
\label{p2}
&\left[(1{-}\hat{\boldsymbol{\beta}}_{i})(\mathbf r_{\max,i}-1)+\frac{1-\xi}{\xi}\frac{1{-}\rho_P}{1{-}\rho_S}\hat{\boldsymbol{\beta}}_{i} \right]e^{- \mathbf r_{\max,i}}
%\nonumber\\&
\leq \mu;
\end{align}

\noindent c) 
Otherwise, $\tilde{\mathbf r}_i(\mu)$ is the only $\mathbf r_i\in[0,\mathbf r_{\max,i}]$ such that 
\begin{align}
\label{p3}
-\left[(1{-}\hat{\boldsymbol{\beta}}_{i})(\mathbf r_{i}-1)+\frac{1-\xi}{\xi}\frac{1{-}\rho_P}{1{-}\rho_S}\hat{\boldsymbol{\beta}}_{i} \right]e^{- \mathbf r_{i}}
+\mu=0.
\end{align}
 $\tilde{\mathbf r}_i(\mu)$ is a non-decreasing function of $\mu$
since, by increasing $\mu$, the  inequality in  (\ref{p1}) becomes tighter
and the  inequality in  (\ref{p2}) becomes looser,
and the left-hand expression of (\ref{p3}) is a decreasing function of $\mathbf r_{i}$.
Hence, $\mu^*(\Lambda)$ is a non-decreasing function of $\Lambda$, so that,
for $\Lambda_1\geq\Lambda_2$,
\begin{align}
\mathbf r^*(\Lambda_1)=
\tilde{\mathbf r}(\mu^*(\Lambda_1))\geq \tilde{\mathbf r}(\mu^*(\Lambda_2))=\mathbf r^*(\Lambda_2).
\end{align}
Property 3) is thus proved.
\\
\emph{Proof of Property 4)}
Let $\mathbf z^*$ be the optimizer of (\ref{optz2}) and let $\hat{\boldsymbol{\beta}}_{i}>\hat{\boldsymbol{\beta}}_{j}$ for some $i\neq j$. Assume by contradiction that
$\mathbf z_i^*>\mathbf z_j^*$.
Now, we define a new SU traffic $\tilde{\mathbf z}$ as follows:
\begin{align}
&\tilde{\mathbf z}_l=\mathbf z_l^*,\ l\notin\{i,j\},\ 
\tilde{\mathbf z}_i=\mathbf z_j^*,
\ 
\tilde{\mathbf z}_j=\mathbf z_i^*.
\end{align}
Equivalently, the SU traffic allocated to the $i$th and $j$th bands under $\mathbf z^*$
is switched under the new traffic scheme $\tilde{\mathbf z}$. Note that
$\sum_i\tilde{\mathbf z}_i=\sum_i{\mathbf z}_i^*=1$,
so that $\tilde{\mathbf z}$ obeys the total traffic constraint, and is thus feasible with respect to (\ref{optz2}).
Let $v(\mathbf z)$ be the value of the objective function in (\ref{optz2}) as a function of $\mathbf z$.
Due to the optimality of $\mathbf z^*$, we have that $v(\tilde{\mathbf z})-v(\mathbf z^*)\leq 0$.
We show that this cannot hold, hence proving the contradiction.
We have
\begin{align}
v(\tilde{\mathbf z})-v(\mathbf z^*)=
(\hat{\boldsymbol{\beta}}_{i}-\hat{\boldsymbol{\beta}}_{j})[\gamma(\mathbf  z_{j}^*)-\gamma(\mathbf  z_{i}^*)],
\end{align}
where we have defined, for $z\in[0,1]$,
\begin{align}
\gamma(z)\triangleq \left[-\Lambda z+\frac{1-\xi}{\xi}\frac{1{-}\rho_P}{1{-}\rho_S} \right]e^{-\Lambda z}.
\end{align}
Note that $\gamma(n)$ is a decreasing function of $n\in[0,1]$. Therefore, since $\mathbf z_i^*>\mathbf z_j^*$,
we have that $\gamma(\mathbf z_i^*)<\gamma(\mathbf z_j^*)$, hence
$v(\tilde{\mathbf z})>v(\mathbf z^*)$, yielding a contradiction.
\vspace{-5mm}
\section*{Appendix C: Proof of Theorem \ref{thm1}}
%In this proof, we consider a generic belief $\pi$, and neglect any functional dependence on $\pi$.
The optimization problem (\ref{P1}) can be decomposed into the following two stages.
First, given $(\bar\beta^{(L)},\bar\beta^{(H)})$ with $\bar\beta^{(L)}\leq\bar\beta^{(H)}$, determine the mapping function $\phi(\cdot)$ such that
\begin{align}
\label{optimalf}
&\phi^*(\bar\beta^{(L)},\bar\beta^{(H)})=\underset{\phi}{\arg\min}\ \mathcal D(\pi,\bar\beta^{(L)},\bar\beta^{(H)},\phi)
\\&\nonumber
=\underset{\phi}{\arg\max}
\sum_i
\left[
{\boldsymbol{\beta}}_{i}\ln \left(\bar\beta^{(\phi(i))}\right)
%\right.\nn\\&\qquad\qquad\left.
+(1-{\boldsymbol{\beta}}_{i})\ln \left(1-\bar\beta^{(\phi(i))}\right)
\right].
\end{align}
Second, determine $\bar\beta^{(L)*}$ and $\bar\beta^{(H)*}$ with optimal mapping $\phi^*(\bar\beta^{(L)},\bar\beta^{(H)})$ into  (\ref{P1}), yielding
\begin{align*}
(\bar\beta^{(L)},\bar\beta^{(H)})^*=\underset{\bar\beta^{(L)},\bar\beta^{(H)}}{\arg\min}\ \mathcal D(\pi,\bar\beta^{(L)},\bar\beta^{(H)},\phi^*(\bar\beta^{(L)},\bar\beta^{(H)})).
\end{align*}
The solution to the intermediate problem (\ref{optimalf}) is trivially given by
\begin{align}
&\phi(i)=L\Leftrightarrow
{\boldsymbol{\beta}}_{i}\ln \left(\bar\beta^{(L)}\right)
{+}(1-{\boldsymbol{\beta}}_{i})\ln \left(1-\bar\beta^{(L)}\right)
%\nonumber\\&
{\geq} {\boldsymbol{\beta}}_{i}\ln \left(\bar\beta^{(H)}\right)
{+}(1-{\boldsymbol{\beta}}_{i})\ln \left(1-\bar\beta^{(H)}\right),
\end{align}
yielding
\begin{align}
\phi(i)=L\Leftrightarrow
{\boldsymbol{\beta}}_{i}
\leq
\left(
1+\frac{\ln \left(\bar\beta^{(H)}\right)-\ln \left(\bar\beta^{(L)}\right)}{\ln \left(1-\bar\beta^{(L)}\right)-\ln \left(1-\bar\beta^{(H)}\right)}
\right)^{-1}.
\end{align}
Note that the solution is of threshold type. Therefore, defining the permutation function $m(\cdot)$ as in the statement of the theorem,
there exists $\nu$ such that $\phi(m(i))=L$ for $i\leq \nu$, and $\phi(m(i))=H$ for $i>\nu$.
We can thus restate the optimization problem (\ref{P1}) by enforcing this solution, yielding
\begin{align}
\label{P2}
(\bar\beta^{(L)*},\bar\beta^{(H)*},\nu^*)
%\\&
%\nonumber
%
=&\underset{\bar\beta^{(L)},\bar\beta^{(H)},\nu}{\arg\max}
\sum_{i=1}^\nu
\left[
{\boldsymbol{\beta}}_{m(i)}\ln \left(\bar\beta^{(L)}\right)
{+}(1{-}{\boldsymbol{\beta}}_{m(i)})\ln \left(1{-}\bar\beta^{(L)}\right)
\right]
%
%\nonumber
\\&
+\sum_{i=\nu+1}^F
\left[
{\boldsymbol{\beta}}_{m(i)}\ln \left(\bar\beta^{(H)}\right)
+(1-{\boldsymbol{\beta}}_{m(i)})\ln \left(1-\bar\beta^{(H)}\right)
\right],
\nonumber
\end{align}
so that $\phi^*(i)=L\Leftrightarrow m(i)\leq \nu^*$.
We solve (\ref{P2}) with respect to $(\bar\beta^{(L)},\bar\beta^{(H)})$ first, for a fixed $\nu$, and then optimize over $\nu$. 
We obtain
\begin{align}
\label{P2b1}
&\bar\beta^{(L)}(\nu)=\underset{\bar\beta^{(L)}}{\arg\max}
\sum_{i=1}^\nu
\left[
{\boldsymbol{\beta}}_{m(i)}\ln \left(\bar\beta^{(L)}\right)
%\right.\\&\left.
+(1-{\boldsymbol{\beta}}_{m(i)})\ln \left(1-\bar\beta^{(L)}\right)
\right]
=
\frac{1}{\nu}\sum_{i=1}^\nu{\boldsymbol{\beta}}_{m(i)},%\nn
\\&
\label{P2b2}
\bar\beta^{(H)}(\nu){=}\underset{\bar\beta^{(H)}}{\arg\max}
\!\!\!\sum_{i=\nu+1}^F\!\!\!
\left[
{\boldsymbol{\beta}}_{m(i)}\ln \left(\bar\beta^{(H)}\right)
%\right.\\&\left.
{+}(1-{\boldsymbol{\beta}}_{m(i)})\ln \left(1-\bar\beta^{(H)}\right)
\right]
{=}
\frac{1}{F-\nu}\sum_{i=\nu+1}^F{\boldsymbol{\beta}}_{m(i)},
%\nn
\end{align}
yielding (\ref{x1}). 
By replacing $(\bar\beta^{(L)}(\nu),\bar\beta^{(H)}(\nu))$ into (\ref{P2}), we finally obtain
$\nu^*$ as in (\ref{nuopt}).
\vspace{-5mm}

%\IEEEtriggeratref{4}
%\enlargethispage{-140mm}
\bibliographystyle{IEEEtran}
\bibliography{IEEEabrv,References}

\end{document}